\documentclass{article}
\usepackage{algorithm}
\usepackage{algpseudocode}
\usepackage{multicol}
\usepackage{pifont}
\usepackage{cite, verbatim}
\usepackage{mathtools}
\usepackage{amssymb, amsmath, mathrsfs, graphicx}
\usepackage{fancyhdr}
\usepackage{float}
\usepackage{tikz}
\usetikzlibrary{matrix, arrows}
\usepackage{graphicx}
\usetikzlibrary{shapes.geometric}
\usetikzlibrary{shapes.arrows}
\usepackage{array}
\usepackage{epstopdf}
\newtheorem{precor}{{\bf Corollary}}

\newenvironment{cor}{\begin{precor}{\hspace{-0.5
em}{\bf.\ }}}{\end{precor}}
\newtheorem{prere}{{\bf Remark }}

 \newtheorem{precon}{{\bf Conjecture}}

\newtheorem{predefin}{{\bf Definition}}

\newenvironment{defin}[1]{\begin{predefin}{\hspace{-0.5
em}{\bf.\ }}{\rm
#1}\hfill{$\blacktriangleleft$}}{\end{predefin}}
\newtheorem{preexm}{{\bf Example}}

\newenvironment{exm}[1]{\begin{preexm}{\hspace{-0.5
em}{\bf.\ }}{\rm #1}\hfill{$\blacktriangleright$}}{\end{preexm}}

\newtheorem{preappl}{{\bf Application}}

\newtheorem{prelem}{{\bf Lemma}}

\newenvironment{lem}{\begin{prelem}{\hspace{-0.5
em}{\bf.\ }}}{\end{prelem}}
\newtheorem{preClaim}{{\bf Claim}}

\newtheorem{preproof}{{\bf Proof.\ }}

\newenvironment{proof}[1]{\begin{preproof}{\rm
#1}\hfill{$\blacksquare$}}{\end{preproof}}
\newtheorem{presproof}{{\bf Sketch of Proof.\ }}

\newtheorem{prethm}{{\bf Theorem}}

\newenvironment{thm}{\begin{prethm}{\hspace{-0.5
em}{\bf.\ }}}{\end{prethm}}

 \newtheorem{prealphthm}{{\bf Theorem}}

\newtheorem{prepro}{{\bf Proposition}}

\newenvironment{pro}{\begin{prepro}{\hspace{-0.5
em}{\bf.\ }}}{\end{prepro}}
\newtheorem{preprb}{{\bf Problem}}

\def\conct[#1, #2]{\mbox {${#1} \leftrightarrow {#2}$}}
\def\dconct[#1, #2]{\mbox {${#1} )=arrow {#2}$}}
\def\deg[#1, #2]{\mbox {$d_{_{#1}}(#2)$}}
\def\mindeg[#1]{\mbox {$\delta_{_{#1}}$}}
\def\maxdeg[#1]{\mbox {$\Delta_{_{#1}}$}}
\def\outdeg[#1, #2]{\mbox {$d_{_{#1}}^{^+}(#2)$}}
\def\minoutdeg[#1]{\mbox {$\delta_{_{#1}}^{^+}$}}
\def\maxoutdeg[#1]{\mbox {$\Delta_{_{#1}}^{^+}$}}
\def\indeg[#1, #2]{\mbox {$d_{_{#1}}^{^-}(#2)$}}
\def\minindeg[#1]{\mbox {$\delta_{_{#1}}^{^-}$}}
\def\maxindeg[#1]{\mbox {$\Delta_{_{#1}}^{^-}$}}
\def\isdef{\mbox {$\ \stackrel{\rm def}{=} \ $}}
\def\dre[#1, #2, #3]{\mbox {${\cal E}_{_{#3}}(#1, #2)$}}
\def\pdre[#1, #2, #3]{\mbox {${\cal P}_{_{#3}}(#1, #2)$}}
\def\var[#1, #2]{\mbox {${\rm Var}_{_{#1}}(#2)$}}
\def\ls[#1]{\mbox {$\xi^{^{#1}}$}}
\def\onvhom[#1, #2]{\mbox {${\rm Hom^{v}}(#1, #2)$}}
\def\onehom[#1, #2]{\mbox {${\rm Hom^{e}}(#1, #2)$}}
\def\core[#1]{\mbox {$#1^{^{\bullet}}$}}
\def\cay[#1, #2]{\mbox {${\rm Cay}({#1}, {#2})$}}
\def\cays[#1, #2]{\mbox {${\rm Cay_{s}}({#1}, {#2})$}}
\def\dirc[#1]{\mbox {$\stackrel{)=arrow}{C}_{_{#1}}$}}
\def\cycl[#1]{\mbox {${\bf Z}_{_{#1}}$}}

\def\sdg[#1]{\mbox {$\stackrel{\leftrightarrow}{#1}$}}

\long\def\Perm[#1]{\mbox{$\lfloor #1 \rceil$}}
\long\def\nat[#1]{\mbox{$\widehat{1...#1}$}}

\newcommand{\matr}[1]{\mathrm{#1}} % For matrices, graphs, ... (structured objects)
\def\Enc{\rm Enc}
\def\Dec{\rm Dec}
\def\Gen{\rm Gen}
\def\Init{\rm Init}

\newcommand{\f}{\mathbb{F}}
\newcommand{\real}{\mathbb{R}}

\newcommand{\xk}{\mathbf{s}( t) }
\newcommand{\uk}{\mathbf{p}( t) }

\newcommand{\kx}{\mathbf{s}( t+1) }
\newcommand{\ke}{\mathbf{e}( t+1)}
\newcommand{\xh}{\widehat{\mathbf{s}}( t)}
\newcommand{\hx}{\widehat{\mathbf{s}}( t+1)}

\newcommand{\hp}{\widehat{\mathbf{p}}}
\newcommand{\uh}{\widehat{\mathbf{p}}( t)}
\newcommand{\ek}{\mathbf{e}( t)}

 \newcommand{\hplainvec}{\widehat{\mathbf{p}}}

\newcommand{\Field}{\mathbb{F}}

\newcommand{\Adv}{{\mathbf{Adv}}}
\newcommand{\Insec}{{\mathbf{Insec}}}
\newcommand{\negl}{{\matr{negl}}}

\newcommand{\outputs}{{\rightarrow}}

\newcommand{\state}{\mathbf{s}}
\newcommand{\hstate}{\widehat{\mathbf{s}}}

\newcommand{\plain}{p}
\newcommand{\cipher}{c}
\newcommand{\keystreamvec}{\mathbf{z}}
\newcommand{\plainvec}{\mathbf{p}}
\newcommand{\ciphervec}{\mathbf{c}}

\newcommand{\memvec}{\mathbf{\mathfrak{m}}}
\newcommand{\key}{\kappa}

\newcommand{\Lmat}{\mathbf{L}}
\newcommand{\Wmat}{\mathbf{W}}

\newcommand{\Bmat}{\mathbf{B}}
\newcommand{\Fmat}{\mathbf{F}}
\newcommand{\Dmat}{\mathbf{D}}
\newcommand{\Emat}{\mathbf{E}}
\newcommand{\Rmat}{\mathbf{R}}
\newcommand{\Amat}{\mathbf{A}}

\newcommand{\memM}{\mathbf{M}}
\newcommand{\BigPi}{\wp}

\newcommand{\encryptionfun}{\varepsilon}
\newcommand{\decryptionfun}{\delta}
\newcommand{\kernel}{\mathrm{Kernel}}

\newcommand{\transmitter}{\mathrm{Enc}}
\newcommand{\receptor}{\mathrm{Dec}}
\newcommand{\Q}{\mathbf{Q}}
\newcommand{\hkeystreamvec}{\mathbf{\widehat{z}}}

 \def\authora{Amir Daneshgar\thanks{Correspondence should be addressed to {\tt daneshgar@sharif.ir}.}}
\def\authorb{Fahimeh Mohebbipoor}

 \def\addressa{
{\small Department of Mathematical Sciences}\\
{\small Sharif University of Technology}\\
{\small P.O. Box {\rm 11155--9415}, Tehran, Iran.}
}

 \def\addressb{
{\small Faculty of Mathematics and Computer Science}\\
{\small Kharazmi University}\\
{\small P.O. Box 15719-14911, Tehran, Iran.}
}

\title{A Secure Self-synchronized Stream Cipher}
\author{
\authora \\
\addressa \\[3mm]
\authorb\\
\addressb\\
}

 \date{}

\begin{document}
\maketitle
\begin{abstract}
\noindent We follow two main objectives in this article. On the one hand, we introduce a security model called LORBACPA$^+$ for self-synchronized stream ciphers which is stronger than the blockwise
LOR-IND-CPA, where we show that standard constructions as delayed CBC or similar existing 
self-synchronized modes of operation are not secure in this stronger model. Then, on the other hand, following contributions of G.~Mill\'{e}rioux et.al., we introduce a new self-synchronized stream cipher and prove its security in LORBACPA$^+$ model.  
\end{abstract}
%----------------------------------------------------section{Introduction}-------------
\section{Introduction}

 Analysis and design of stream ciphers is a classic and among the oldest subjects in
cryptography, however, amazingly, there are still some challenging problems to be addressed by the experts in the field (e.g. see \cite{Pre08,DK08}).

 Among these interesting and challenging problems one may recall the analysis and design of self-synchronized stream ciphers where despite the efforts made so far, there does not exist a deep understanding of design methods, analysis and security models for this kind of stream ciphers yet.

\subsection{Main results}

 Although, it is hopeless to think of CCA secure self-synchronized stream ciphers
because of their error-correction properties, we show that it is possible to design such systems which are secure in quite stronger models than the classical CPA setting.
Our main contribution in this article is to propose such a self-synchronized stream cipher along with the security model.

 Before we proceed, it is instructive to note that our results are based on three fundamental contributions in cryptography.

 On the one hand, not only our security model are built on the basic contribution of
Bellare, Desai, Jokipii and Rogaway in \cite{BDJR97} who introduced the concept of
{\it left-or-right indistinguishability} (${\rm LOR-IND} $) but also we adopt their fundamental approach for the method of security proof in Section~\ref{sec:cryptoanalysis}.

 On the other hand, we will basically adopt the notion of {\it blockwise security} introduced independently by
Joux, Martinet and Valette \cite{JMV02} and Bellare, Kohno and Namprempre \cite{BKN04} (also see \cite{BR05,FJP04,FJMV03,Jou10} for the background) since our self-synchronized stream cipher
can be considered as an extension of modes of operation for block-ciphers.

 Although, the adaption of these ideas are crucial in our contribution, but our system
would not be secure without applying basic ideas from what G.~Mill\'{e}rioux et.al.
has done on connecting the design principles of self-synchronized stream cipher to
basic concepts of synchronized control systems (see \cite{PGM11},\cite{PM13} and references therein).

 The canonical form of a Self-Synchronizing
Stream Cipher (an SSSC in short) is made of a combination of a shift register, which
acts as a state register with the ciphertext as input, together with a filtering
function that provides the running key stream and an output function which combines the running key stream with plaintext to produce the cipher text. In particular, it is known that
such canonical SSSC's are ${\rm IND-CPA}$ secure under some conditions on the filtering function \cite{BDM15,DGM16}.

 As a matter of fact, using the control-theoretic approach of G.~Mill\'{e}rioux et.al., we have been able to make sure about the existence of {\it free random initial states} in our self-synchronized stream cipher
that will guarantee the security of the system in a stronger setting than the CPA model
traditionally used for modes of operation. These basic ideas will give rise to our proposed system introduced in Section~\ref{sec:DSPLCDESCRP}. Also, we will discuss some more control-theoretic aspects of our designs in Section~\ref{sec:conclusion}.

However, to begin, after covering the necessary background in the rest of this section
we will concentrate on the details of the proposed security model (i.e. the LORBACPA$^+$ setup) in Section~\ref{sec:Sm}. Also, in this section we will elaborate on mentioning the basic ideas and problems using the classical CBC and CFB modes of operation and will mention two premature modified versions of these modes that will both illustrate the basic design principles on the one hand, and shows that the basic CBC and CFB modes are not flexible enough to lead to a fully secure system in our security model. This in conjunction with the control theoretic approach of G.~Mill\'{e}rioux et.al. shows the importance of handling
the initial vectors to gain maximal security and operational efficiency in the design of stream ciphers.
 The main setup of our proposed self-synchronized stream cipher will be introduced in Section~\ref{sec:DSPLCDESCRP} where we also present the full security proof in our security model in Section~\ref{sec:cryptoanalysis}.

As for the notations,
the symbol $\Field_{_{q}}=GF(q)$ for $q$ a prime power, stands for the finite field on $q$ elements, $\real$ is the field of real numbers, and
$[n]$ stand for the set $\lbrace1,2 ..., n\rbrace$, respectively. Also, if $b \in \{0,1\}$ then $\bar{b} \in \{0,1\}$ is the other complement bit, i.e, $\bar{b} \isdef b-1$ modulo $2$. 

 A vector of $n$ elements in $\Field_{_{q}}^{n}$ is denoted by
$\mathbf{s}=(s_{1}, s_{2}, \cdots, s_{n}) ^{T}$.
If the elements of a vector $\mathbf{s}$ are functions of a variable $t$, then we write
$\xk$ to refer to the vector $\mathbf{s}$ at time $t$.
The set of $n\times n $ matrices with entries in $\Field_{_{q}}$ is denoted by $\mathcal{M}^{n \times n}$. The symbol $\mathbf{I}$ stands for the identity matrix, and
$\mathbf{{\texttt{0}}}$
stands for the zero matrix (when the dimension is clear from the context).
The symbol $\state \overset{\$}{\leftarrow} \Field_{_{q}}^{n}$ demonstrates the action of picking $\state$ uniformly at random from the set $\Field_{_{q}}^{n}$.

 \subsection{Self-synchronized stream-ciphers}\label{sec:basics}

 From a data-transmission point of view a self-synchronized stream cipher
is a master-slave communication system consisting of a transmitter and a receiver as

 \begin{equation}\label{Y8}
\begin{array}
{ll}
\Sigma_{_{\theta}}:&
\left\lbrace \begin{array}
{l}
\mathbf{z}(t)=f_{_{\theta}}(\mathbf{c}(t-l),\cdots,\mathbf{c}( t-l')), \\
\mathbf{c}(t+d)=enc_{_{\theta}}(\mathbf{z}(t),\uk),
\end{array}
\right.
\\ \ \\
\Sigma'_{_{\theta}}:&
\left\lbrace \begin{array}
{l}
\widehat{\mathbf{z}}(t+d)=f_{_{\theta}}(\mathbf{c}(t-l),\cdots,\mathbf{c}( t-l')), \\
\widehat{\mathbf{p}}(t+d)=dec_{_{\theta}}( \widehat{\mathbf{z}}(t+d),\mathbf{c}(t+d)),
\end{array}
\right.
\end{array}
\end{equation}
where $f_{_{\theta}}$ is the function that generate the key-streams $\{\mathbf{z}(t)\}$ and $\widehat{\mathbf{z}}$. The ciphertext $\mathbf{c}(t+d)$ is worked out through an encryption function $enc_{_{\theta}}$ and the
decryption is performed through a function $dec_{_{\theta}}$ depending on the ciphertext $\mathbf{c}(t+d)$, in which $l' \geq l \geq 1$ and $d\geq 0$ is the system delay.
Note that if $dec_{_{\theta}}=enc_{_{\theta}}^{-1}$ and for some finite time synchronization delay $t_{_{s}} > d$, we have
$$\forall \ t \geq t_{_{s}}, \ \ \widehat{\mathbf{z}}(t+d)=\mathbf{z}(t),$$
then
$$\forall \ t \geq t_{_{s}}\geq 1 \ \ \widehat{\mathbf{p}}(t+d)=dec_{_{\theta}}( \widehat{\mathbf{z}}(t+d),\mathbf{c}(t+d))=\uk.$$
In this setting, $t_{_{c}} \isdef |l'-l+1|$ shows the amount of memory one needs to save the necessary ciphertexts from the past.
Hence, the self-synchronizing stream cipher must be initialized by loading $t_{_{c}}$ dummy ciphertext symbols at the beginning as part of the {\it initial-condition vector} ICV.
Note that to ensure correct deciphering of the plaintext, this data must be shared between the transmitter and the receiver but does not necessarily need to be kept secret.

 Let us emphasize that in Equation~\ref{Y8} and in all other equations that describe stream ciphers in the sequel, $t \in \mathbb{Z}$ denotes the
time, $\mathbf{c}(t)$ stands for the cipher stream (in blocks) that is transmitted on the communication channel between the transmitter $\Sigma_{_{\theta}}$ and the receiver $\Sigma'_{_{\theta}}$. Similarly, $\mathbf{p}(t)$ stands for the
plaintext stream (in blocks) and $\widehat{\mathbf{p}}(t)$ stands for the stream (in blocks) that is extracted by the receiver. Hence, in this setting, for instance,
equations of $\Sigma_{_{\theta}}$ imply that $\mathbf{c}(t+d)$ is generated
at time $t$ using the data $\mathbf{z}(t)$ and $\uk$, but is transmitted on the channel $d$ clocks later, at time $t+d$. On the other hand, the equations of $\Sigma'_{_{\theta}}$ show that $\widehat{\mathbf{p}}(t+d)$ is computed
at time $t+d$ using $\widehat{\mathbf{z}}(t+d)$ and $\mathbf{c}(t+d)$ that is received through the communication channel with a delay $d$. These facts imply that, generically, in a self-synchronized setting, and when both systems are synchronized, we have $\widehat{\mathbf{z}}(t+d)=\mathbf{z}(t)$ and $\widehat{\mathbf{p}}(t+d)=\mathbf{p}(t)$ for $t\geqslant 1$. In what follows we always assume that the delayed signals $\mathbf{z}(t)$ and $\widehat{\mathbf{p}}(t)$ are set to $\bot$ when they are not defined.

 A study of stream ciphers through a control theoretic approach has been pioneered
by G.~Mill\'{e}rioux {\it et.al.} (see \cite{PM13} and references therein) from which we know that the setup introduced in
Equations~\ref{Y8} is equivalent to the following setup under {\it flatness} condition (see Lemma~\ref{lem:flatness}),

 \begin{equation}\label{eq:ss}
\begin{array}{ll}
\Gamma_{_{\theta}}:&
\left\lbrace \begin{array}{l}
\mathbf{s}(t+1)=\phi_{_{\theta}}(\xk,\uk), \\
\mathbf{c}(t+d)=\varepsilon_{_{\theta}}(\xk,\uk),
\end{array}
\right.
\\ \ \\
\Gamma'_{_{\theta}}:&
\left\lbrace \begin{array}{l}
\widehat{\mathbf{s}}(t+1)=\beta_{_{\theta}}(\xh,\mathbf{c}(t)), \\
\widehat{\mathbf{p}}(t)=\delta_{_{\theta}}(\xh,\mathbf{c}(t)),
\end{array}
\right.
\end{array}
\end{equation}
where $\phi_{_{\theta}}$ and $\beta_{_{\theta}}$ are the functions that generate $\xk$ and $\xh$ and also $\varepsilon_{_{\theta}}$, $\delta_{_{\theta}}$  are the encryption and the decryption functions.
 Such a system is said to be {\it finite-time self-synchronized} if
there exists some finite time $t_{_{s}}$
such that for any initial condition
$$ \forall t \geq t_{_{s}}, \ \ \widehat{\mathbf{s}}(t+d)=\xk.$$
Note that after synchronization, i.e. for all $t \geq t_{_{s}}$ we have
$$\hp(t+d)=\plainvec(t).$$

 Although systems~\ref{Y8} and \ref{eq:ss} are essentially equivalent under flatness condition (see Lemma~\ref{lem:flatness}), but as our first crucial observation, we will see that there is a subtlety about the mapping between initial conditions. In other words, the definition of {\it initialization vector}, IV, as it is usually referred to in a cryptographic context is based on initializing System~\ref{Y8}, while for System~\ref{eq:ss}
one may choose parts of the initial conditions randomly and still recover
the plaintex correctly using self-synchronization. This important
technique of mapping self-synchronization to a random initial condition
is fundamental in our security analysis and shows the importance of the
approach proposed by G.~Mill\'{e}rioux {\it et.al.} for the design of
self-synchronized stream ciphers which is based on the design of flat
systems of type~\ref{eq:ss}.

 Since we are going to analyze our proposed systems in a provable security model, we elaborate on emphasizing on the computational details of our model as follows.

 A self-synchronized dynamical cryptosystem scheme $$\Gamma_{_{d,t_{_{s}}}}=(\Gen,\Init_{_{\Enc}},
\Init_{_{\Dec}},\Enc_{_{\theta}},\Dec_{_{\theta}})$$ is defined to be a set of efficient algorithms $\Gen$, $\Init_{_{\Enc}}$, $\Init_{_{\Dec}}$, $\Enc_{_{\theta}} $ and $\Dec_{_{\theta}} $ described as follows.
\begin{itemize}
	\item[-] $k$: is the security parameter,
	\item[-] $\Gen$: is a probabilistic algorithm that on input $1^k$
	outputs the secret key $\key$.
	\item[-] $\theta$ is a vector containing system parameters determined
	 by $\key$ and the public information.
	\item[-] $\Init_{_{\Enc}}$: is a randomized algorithm that on input $\key$ outputs an initial condition vector $\matr{ICV}$ and the state
	initialization vector $\state(0)$.
\item[-] $\Init_{_{\Dec}}$: is a randomized algorithm that on inputs $\key$ and the vector $\matr{IV}$ (received through the communication channel) outputs an initial condition vector $\matr{ICV}$ and the state
	initialization vector $\hstate(d)$.
	\item[-] $\Enc_{_{\theta}}$: is a randomized algorithm that using initializations, a message block $\uk$ and a state vector $\xk$, outputs a ciphertext block $\mathbf{c}(t)$ and an	updated state vector $\kx$ as follows,
	
\begin{equation}\label{eq:sssen}
\Enc_{_{\theta}}(t):
\left\lbrace \begin{array}{l}
(\state(0),\matr{ICV}) \leftarrow \Init_{_{\Enc}},\\
If \ t=0 \ then \\
\ \ \mathbf{s}(1)=\phi'_{_{\theta}}(\state(0),\matr{ICV}), \\
\ \ \mathbf{c}(d)=\varepsilon'_{_{\theta}}(\state(0),\matr{ICV}),\\
\ \ {\rm Output:} \ \bot\\
Else\\
\ \ \mathbf{s}(t+1)=\phi_{_{\theta}}(\xk,\uk), \\
\ \ \mathbf{c}(t+d)=\varepsilon_{_{\theta}}(\xk,\uk),\\
\ \ {\rm Output:} \ \mathbf{c}(t),
\end{array}
\right.
\end{equation}
where $\phi'_{_{\theta}}$ and $\phi_{_{\theta}}$ are the functions that generate $\mathbf{s}(1)$ and $\mathbf{s}(t+1)$ and also $\varepsilon'_{_{\theta}}$ and $\varepsilon_{_{\theta}}$ are encryption functions.
 	\item[]$\Dec_{_{\theta}}$: is a randomized algorithm that
	using initializations, a ciphertext $\mathbf{c}(t) $ and a state vector $\xh$, outputs a (recovered plaintext) block $\uh$ and an	updated state vector $\hx$ as follows,

 \begin{equation}\label{eq:sssde}
\Dec_{_{\theta}}(t):
\left\lbrace \begin{array}{l}
(\hstate(d),\matr{ICV}) \leftarrow \Init_{_{\Dec}},\\
If \ t < d\ then \\
\ \ {\rm Output:} \ Ack \\
If \ t = d\ then \\
\ \ \widehat{\mathbf{s}}(t+1)=\beta'_{_{\theta}}(\xh,\mathbf{c}(d),\matr{ICV}), \\
\ \ {\rm Output:} \ Ack \\
Else\\
\ \ \widehat{\mathbf{s}}(t+1)=\beta_{_{\theta}}(\xh,\mathbf{c}(t)), \\
\ \ \widehat{\mathbf{p}}(t)=\delta_{_{\theta}}(\xh,\mathbf{c}(t)),\\
\ \ {\rm Output:} \ \widehat{\mathbf{p}}(t).
\end{array}
\right.
\end{equation}
Where $\beta'_{_{\theta}}$ and $\beta_{_{\theta}}$ are the functions that generate $\widehat{\mathbf{s}}(1)$ and $\widehat{\mathbf{s}}(t+1)$ and also $\delta_{_{\theta}}$ is decryption functions.
 \end{itemize}

 Here again $d$ is the system delay and $t_{_{s}} > d$ is the synchronization delay, respectively. We usually assume that $\delta_{_{\theta}}=\varepsilon_{_{\theta}}^{-1}$ and that the first $(t_{_{s}}-1)$ plaintexts are selected randomly to ensure safe decryption of the plaintext. As a natural condition we always assume that $\varepsilon_{_{\theta}}$ resists collisions in the sense that for any pair of states $(\mathbf{s},\mathbf{s}')$, the probability of  not having a collision like
$$\varepsilon_{_{\theta}}(\mathbf{s},\mathbf{p})=\varepsilon_{_{\theta}}(\mathbf{s}',\mathbf{p}')$$
for some $\mathbf{p} \not = \mathbf{p}'$ is greater than a noticeable function of the security parameter.

 Also, note that the initial condition vector
is generated using the secret key $\key$, however the whole vector ICV
is not necessarily transmitted to the receiver on the transmission channel along with the ciphertext. For this, in what follows, IV always
stands for part of ICV that is transmitted (plainly) along with the ciphertext to the receiver. Note that in this setting the rest of ICV is either determined by the secret key $\key$ or is chosen at random (of course the IV itself can also depend on the $\key$ or be chosen at random but the difference falls in the fact that IV is transmitted through the channel to the receiver).

 The parametric function $\phi_{_{\theta}}$ is called the
{\it update} of the encryption process $\transmitter$, and the
parametric function $\beta_{_{\theta}}$ is called the {\it update} of the decryption process $\receptor$.

 The following lemma essentially reflects the main contribution of
G.~Mill\'{e}rioux {\it et.al.} (see \cite{MD09} and references therein)
expressed in our setting.

 \begin{lem}\label{lem:flatness}
Consider the communication system $\Gamma_{_{d,t_{_{s}}}}$. Assuming that for all $ t \geq t_{_{s}}$ we have $\widehat{\mathbf{s}}(t+d)=\mathbf{s}(t)$, then,
\begin{itemize}
\item[{\rm i)}]If $\delta_{_{\theta}}=\varepsilon_{_{\theta}}^{-1}$
then for all $ t \geq t_{_{s}}$ we have $ \widehat{\mathbf{p}}(t+d)=\plainvec(t)$.
\item[{\rm ii)}] If $\Enc_{_{\theta}}$ is flat, i.e.
if there exist constants $l,l'$ and functions $F_{_{\theta}}$ and $G_{_{\theta}}$ such that
\begin{equation}\label{eq:flat}
\Gamma_{_{d,t_{_{s}}}}:
\left\lbrace\begin{array}{l}
\xk=F_{_{\theta}}(\ciphervec(t-l),\cdots,\ciphervec(t-l')),\\
\uk=G_{_{\theta}}(\ciphervec(t-l),\cdots,\ciphervec(t-l')),
\end{array}
\right.
\end{equation}

 then $\Gamma_{_{d,t_{_{s}}}}$ is equivalent to the following system which is in the form of System~{\rm \ref{Y8}} $($note that the converse is trivial$)$,

 \begin{equation}\label{eq:sssen1}
\Sigma_{_{\theta}}(t):
\left\lbrace \begin{array}{l}
Initialize \ (\ciphervec(t_{_{s}}-l),\cdots,\ciphervec(t_{_{s}}-l')), \\
If \ t\geq t_{_{s}} \ then \\
\ \ \mathbf{z}(t)=H_{_{\theta}}(\ciphervec(t-l),\cdots,\ciphervec(t-l')), \\
\ \ \mathbf{c}(t+d)=\varepsilon_{_{\theta}}(\mathbf{z}(t),\uk).\\
%\ \ {\rm Output:} \ \mathbf{c}(t).
\end{array}
\right.
\end{equation}
\ \\
\begin{equation}\label{eq:sssde1}
\Sigma'_{_{\theta}}(t):
\left\lbrace \begin{array}{l}
Receive \ (\ciphervec(t_{_{s}}-l),\cdots,\ciphervec(t_{_{s}}-l')),\\
If \ t\geq t_{_{s}} \ then \\
\ \ \widehat{\mathbf{z}}(t+d)=H_{_{\theta}}(\ciphervec(t-l),\cdots,\ciphervec(t-l')),\\
\ \ \widehat{\mathbf{p}}(t+d)=\delta_{_{\theta}}(\widehat{\mathbf{z}}(t+d),\mathbf{c}(t+d)).\\
%\ \ {\rm Output:} \ \widehat{\mathbf{p}}(t+d).
\end{array}
\right.
\end{equation}

 \end{itemize}
\end{lem}

 Nowadays, self-synchronized stream ciphers (at least the classic ones) are much more related to block ciphers than to synchronous stream ciphers. In this regard, it is well-known that several modes of operation for block ciphers have been proposed
among which a couple of them (as $\matr{CBC}$ and $\matr{CFB}$ modes) are self-synchronized (e.g. see \cite{oS80}).

 \begin{exm}{{\bf CBC and CFB modes}
		
		Given a block cipher with encryption $E_{\key}$ and decryption $D_{\key}$,	we may describe the $\matr{CBC}$ and $\matr{CFB}$ modes as follows.
	Note that in both cases $\Init_{_{\Enc}}$ and $\Init_{_{\Dec}}$ are randomized algorithms that on input $1^k$,
	output a random initial value $\matr{IV}$ (here $t_{_{s}}=t_{_{c}}=1$ and $d=0$).
	\newpage
	The $\matr{CBC}$ mode:
	\begin{itemize}
			\item[]
	{\scriptsize
		\begin{equation}\label{esssc}
			\begin{array}{ll}
				\Enc_{_{\theta}}(t):
				\left\lbrace \begin{array}{l}				
	 Initialize \ 	\matr{ IV},\\
	 	 \state(0) =\matr{ IV},\\
					If \ t=0 \ then \\
					\ \ \state(1) = \state(0),\\
					\ \ {\rm Output:} \bot \\
					Else \\
					\ \ \mathbf{c}(t) = E_{\key}(\uk \oplus\xk),\\
					\ \ \kx = \mathbf{c}(t), \\
					\ \ {\rm Output:} \mathbf{c}(t). \\
				\end{array}
				\right.
				&
				\Dec_{_{\theta}}(t):
				\left\lbrace \begin{array}{l}		
			 	Receive \ \matr{ IV},\\
			 	\hstate(0) =\matr{ IV},\\
					If \ t=0 \ then \\
					\ \ \hstate(1) =\hstate(0), \\
					\ \ {\rm Output:} Ack \\
					Else\\
					\ \ \hx = \mathbf{c}(t), \\
					\ \ \uh=D_{\key}(\mathbf{c}(t))\oplus\xh,\\
					\ \ {\rm Output:} \uh.\\
				\end{array}
				\right.
			\end{array}
		\end{equation}
	}
	\end{itemize}
	The $\matr{CFB}$ mode:
	\begin{itemize}
			\item[]
	{\scriptsize
		\begin{equation}\label{esssc1}
			\begin{array}{ll}
				\Enc_{_{\theta}}(t):
				\left\lbrace \begin{array}{l}
				{\rm Initialize }\ \matr{ IV},\\
				\state(0) =\matr{ IV},\\
					If \ t=0 \ then \\
					\ \ \state(1) = E_{\key}(\matr{ IV}),\\
					\ \ {\rm Output:} \bot \\
					Else\\
					\ \ \kx = E_{\key}(\mathbf{c}(t) ),\\
					\ \ \mathbf{c}(t) = \uk\oplus \xk,\\
							\ \ {\rm Output:}  \mathbf{c}(t).\\
				\end{array}
				\right.
				&
				\Dec_{_{\theta}}(t):
				\left\lbrace \begin{array}{l}
				{\rm Receive }\ \matr{ IV},\\
				\hstate(0) =\matr{ IV},\\
					If \ t=0 \ then \\
					\ \ \hstate(1) = E_{\key}(\matr{ IV}),\\
					\ \ {\rm Output:} Ack \\
					Else\\
					\ \ \hx = E_{\key}(\mathbf{c}(t)),\\
					\ \ \uh=\xh\oplus\mathbf{c}(t),\\
						\ \ {\rm Output:}  \uh. \\
				\end{array}
				\right.
			\end{array}
		\end{equation}
	}
	\end{itemize}

}\end{exm}

\section{Security models}\label{sec:Sm}

 In \cite{JMV02} Joux et.al. introduced a new attack to some modes of operation including the $\matr{CBC}$ mode which was the initiating fundamental contribution leading to the {\it blockwise}
models of security for stream ciphers (also see \cite{BKN04} for other motivations). In \cite{FMP03} the authors proposed a simple delay procedure to prevent the proposed attack as follows.

 \begin{exm}{\label{exm:DCBC}
		{\bf DCBC: the delayed $\matr{CBC}$ mode \cite{FMP03}}
	
		\begin{itemize}
			\item[-] Here $\Init_{_{\Enc}}$ and $\Init_{_{\Dec}}$ are randomized algorithms that on input $1^k$,
			output a random initial value $\matr{IV}$ (note that here we have $t_{_{s}}=t_{_{c}}=d=1$).
			\item[-] $\Enc_{_{\theta}}(t)$
			takes as inputs a plaintext block or the special symbol ``stop" for any time.
	{\scriptsize
		\begin{equation}\label{DCBC}
			\begin{array}{ll}
				\Enc_{_{\theta}}(t):
				\left\lbrace \begin{array}{l}
								\rm Initialize \ \matr{ IV},\\
					\state(0) =\matr{ IV},\\
					If \ t=0 \ then\\
					\ \ \mathbf{c}(1) =\state(0),\\
					\ \ \state(1) =\state(0), \\
					\ \ {\rm Output:} \ \bot\\
						 If \ \plainvec(t)=stop \ then \\
					\ \ {\rm Output:}\ \ciphervec(t)\\
					Else \\
					\ \ \mathbf{c}(t+1) = E_{\key}(\state(t) \oplus \plainvec(t)),\\
					\ \ \state(t+1) = E_{\key}(\state(t) \oplus \plainvec(t)),\\
					\ \ {\rm Output:}\ \mathbf{c}(t).\\
				
				\end{array}
				\right.
				&
				\Dec_{_{\theta}}(t):
				\left\lbrace \begin{array}{l}
				{\rm Receive }\ \matr{ IV},\\
				\hstate(1) =\matr{ IV},\\		
				If \ t=1 \ then \\
					\ \ \hstate(2)=\matr{ IV},\\
					\ \ \hp(1)=\ {\rm Ack},\\
					\ \ {\rm Output:} \hp(1)\\
					Else \\
					\ \ \hstate(t+1) = \mathbf{c}(t),\\
					\ \ \hp(t) = D_{\key}(\mathbf{c}(t)) \oplus \hstate(t),\\
					\ \ {\rm Output:}\ \hp(t).\\
				\end{array}
				\right.
			\end{array}
		\end{equation}
	}
%	where $\hstate(t+1) =\state(t)$ for any $t>1$.
	\end{itemize}
			
}\end{exm}
	
Although the proposed $\matr{DCBC}$ mode is secure in the blockwise security model (the details will follow) let us consider the following setup showing that the delay procedure also has some drawbacks.

 Assume that the adversary has oracle access to the decryption update (not necessarily to the whole decryption procedure) and uses the decryption
update to do the encyption in the following setting.
	{\scriptsize
\begin{equation}\label{eq:DCBCEn'}
\Enc'_{_{\theta}}(t):
\left\lbrace \begin{array}{l}
				{\rm Initialize }\ \matr{ IV},\\
	\ \hstate(1) =\matr{ IV},\\		
If \ t=1\ then \\
\ \	\mathbf{c}(1)=\matr{IV},\\
\ \ \hstate(2)=\matr{IV}, \\
\ \ {\rm Output:} \ \mathbf{c}(1).\\
If\ \plainvec(t-1)=stop \ then \\
\ \ {\rm Output:}\ \bot\\
Else\\	
\ \	\mathbf{c}(t)=E_{\key}(\xh\oplus \plainvec(t-1)),\\
\ \ \hstate(t+1)=\mathbf{c}(t), \\

 \ \ {\rm Output:} \ \mathbf{c}(t).\\
 
\end{array}
\right.
\end{equation}
}

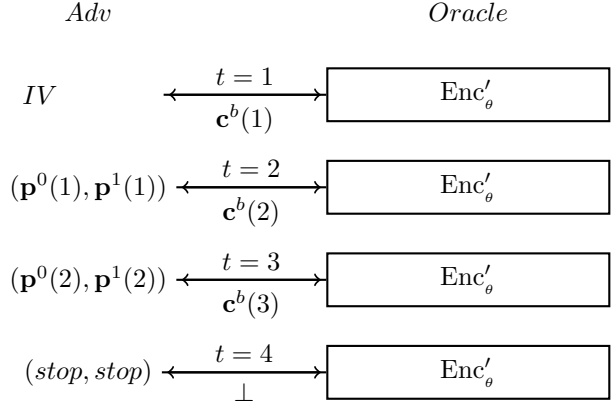
\begin{figure}[t]
	\centering
	\begin{tikzpicture} [
	auto,
	decision/.style = { diamond, draw=black, thick,
		text width=5em, text badly centered,
		inner sep=1pt },
	block/.style = { rectangle, draw=black, thick,
		text width=10em, text centered,
		minimum height=2em },
	line/.style = { draw, thick, ->, shorten >=2pt },
	]
	
	% Define nodes in a matrix
	\matrix [column sep=20mm, row sep=5mm]
	{\node [text centered] (x0) {$Adv $};&\node [text centered] (p0) {$Oracle$};\\
		\node [text centered] (x1) {$IV~~~~~~~ ~~~~$};&\node [block] (p1) {$\Enc'_{_{\theta}}$};\\
		\node [text centered] (x2) {$(\plainvec^0(1),\plainvec^1(1))$  };&\node [block] (p2) {$\Enc'_{_{\theta}}$};\\
		\node [text centered] (x3) {$(\plainvec^0(2),\plainvec^1(2))$ };&\node [block] (p3) {$\Enc'_{_{\theta}}$}; \\
		\node [text centered] (x4) {$(stop,stop)$ };&\node [block] (p4) {$\Enc'_{_{\theta}}$};\\
	};
	
	%
	% legend for subprocedures
	\begin{scope} [every path/.style=line]
	\path (x1) -- node [] {$t=1$} (p1);
	\path (p1) -- node [] {$\mathbf{c}^b(1)$} (x1);
	\path (x2) -- node [] {$t=2$} (p2);
	\path (p2) -- node [] {$\mathbf{c}^b(2)$} (x2);
	\path (x3) -- node [] {$t=3$} (p3);
	\path (p3) -- node [] {$\mathbf{c}^b(3)$} (x3);
	\path (x4) -- node [] {$t=4$} (p4);
	\path (p4) -- node [] {$\bot$} (x4);

	\end{scope}
	\end{tikzpicture}
	\caption{ The interaction between the adversary and the oracle $\Enc'_{_{\theta}}$.
	}\label{fig:DCBC}
\end{figure}

 Note that in a synchronized regime the response of $\Enc'_{_{\theta}}$
is the same as that of $\Enc_{_{\theta}}$.
The following attack using the encryption oracle
$\Enc'_{_{\theta}}$ succeeds in a distinguishability test as it is described below.
This kind of attack is our basic motivation for our propose security model in Section~\ref{sec:LORBACPA}.
\begin{itemize}
\item[-]The adversary uniformly chooses two messages $\left\lbrace\plainvec^0(t)\right\rbrace, \left\lbrace\plainvec^1(t)\right\rbrace,$ each consisting of two blocks, such that
$\plainvec^0(2)\neq \plainvec^1(2) $, and queries its oracle on this pair.
\item[-] The encryption of $\left\lbrace\plainvec^b(t)\right\rbrace$ as $(\ciphervec^{b}(1),\ciphervec^{b}(2),\ciphervec^{b}(3),\matr{IV})$ is
given to the adversary. (The goal of the adversary will be to guess the value of $b$.)
\item[-] The adversary queries a pair $(\tilde{\plainvec}^0(1),\tilde{\plainvec}^1(1))$
for plaintexts $(\left\lbrace\tilde{\plainvec}^0(t)\right\rbrace,\left\lbrace\tilde{\plainvec}^1(t)\right\rbrace)$ whose first blocks, $\tilde{\plainvec}^0(1)$ and $\tilde{\plainvec}^1(1)$ are chosen uniformly at random.
\item[-] The adversary receives the encryption of $\tilde{\plainvec}^b(1)$ as $(\tilde{\ciphervec}^b(1),\tilde{\ciphervec}^b(2))$ and then sends the second query as $(\tilde{\plainvec}^0(2),\tilde{\plainvec}^1(2))$
for which
$\tilde{\plainvec}^0(2)=\tilde{\plainvec}^1(2)=\plainvec^0(2) \oplus \ciphervec^{b}(2) \oplus \tilde{\ciphervec}^{b}(2)$.
\item[-]The adversary receives $\tilde{\ciphervec}^b(3)$.
\item[-] If $\tilde{\ciphervec}^b(3)=\ciphervec^{b}(3)$ the adversary guesses $b = 0$, otherwise it sets $b = 1$.
\end{itemize}

 Note that if $b = 0$ then,
\[\begin{array}{ll}
\tilde{\ciphervec}^b(3)=E_{\key}({\hstate}^b(3) \oplus \tilde{\plainvec}^b(2))=
E_{\key}(\tilde{\ciphervec}^b(2) \oplus \plainvec^0(2)\oplus \ciphervec^{b}(2)\oplus \tilde{\ciphervec}^b(2))=\ciphervec^{b}(3),
\end{array}
\]
while if $b = 1$ with a chance of at least $\frac{1}{2}$ we have,
\[\begin{array}{ll}
\tilde{\ciphervec}^b(3)=E_{\key}({\hstate}^b(3) \oplus \tilde{\plainvec}^b(2))=
E_{\key}(\tilde{\ciphervec}^b(2) \oplus \plainvec^0(2)\oplus \ciphervec^{b}(2)\oplus \tilde{\ciphervec}^{b}(2)) \neq \ciphervec^{b}(3).
\end{array}
\]

 It is instructive to note that a similar attack is not applicable using the $\Enc_{_{\theta}}$ oracle. Also, this example shows that although applying a delay procedure may
prevent attacks in the blockwise security model but at the same time it may still make the whole system vulnerable to the attacks using the decryption update (which is quite feasible in practice).

 \subsection{The LORBACPA$^+$ security model}\label{sec:LORBACPA}

 Our main objective in this section is to formalize a security model that
can be used for both random and nonrandom (nonce) initialization scenarios within the blockwise model using decryption update. It will become clear that this security model, hereafter called LORBACPA$^+$, is stronger than the blockwise CPA model,
where our ultimate aim in the next sections will be to propose a LORBACPA$^+$ secure self-synchronized stream cipher.

 To begin, we have to formalize the oracles we are going to use in our security model.

 \begin{defin}{{\bf LORBA encryption oracles}\ \\
		
		A {\it left-or-right blockwise adaptive encryption oracle} $E(-,b)$ is an oracle
		applying the function $\Enc_{_{\theta}}$ in the following way (depending on the
		specified parameters).
		\begin{itemize}
			\item Each communication with the oracle is either an {\it initialization}
			request of a new {\it session} or a request for {\it encryption}.
			\item For each initialization request we have the following cases:
			 \begin{itemize}
			 	\item If the oracle is a randomly initialized oracle, denoted by $E({\rm \$IV},b)$, then the oracle
			 	 starts a new session using a random initializing data and sends the number of the session along with the (randomly fixed) $\matr{IV}$ to the main program.
			 	\item If the oracle is of chosen $\matr{IV}$ type, denoted by $E(\matr{IV},b)$, then the initialization request contains an $\matr{IV}$ chosen by the program (i.e. the adversary)
			 	 and the oracle uses this $\matr{IV}$ along with the key and possibly other randomly chosen initializing parameters (if any) and returns the session number to the main program.
			 	\item The oracle uses $\Init_{_{\Enc}}$ to determine $\state(0)$.
			 	 \end{itemize}
			 	 \item Each encryption request is of the form $(\plainvec^0,\plainvec^1,i)$ containing two plaintext-blocks
			 $\plainvec^0$ and $\plainvec^1$ with same length and session number $i$.
		 The oracle is capable of saving the history of each session
		 (in particular the state vectors $\state(t)$), therefore, upon
			 receiving such a request the oracle returns the encryption of the block $\plainvec^b$ using the history of the session $i$ along with the $\matr{IV}$ by applying
			
\begin{equation}\label{eq:sssen3}
\Enc_{_{\theta}}(t):
\left\lbrace \begin{array}{l}
{\rm Use \ history \ of }\ (\state(0), \matr{IV}, \matr{ICV}),\\
If \ t=0 \ then \\
\ \ \mathbf{s}(1)=\phi'_{_{\theta}}(\state(0),\matr{ICV}), \\
\ \ \mathbf{c}(d)=\varepsilon'_{_{\theta}}(\state(0),\matr{ICV}),\\
\ \ {\rm Output:} \ \mathbf{c}(0)\\
Else\\
\ \ \mathbf{s}(t+1)=\phi_{_{\theta}}(\xk,\plainvec^b(t)), \\
\ \ \mathbf{c}(t+d)=\varepsilon_{_{\theta}}(\xk,\plainvec^b(t)),\\
\ \ {\rm Output:} \ \mathbf{c}(t).
\end{array}
\right.
\end{equation}
		
Note that the oracle is capable of answering different queries
related to different sessions intermittently.
			
\end{itemize}
		
}\end{defin}

 \begin{defin}{{\bf Synchronized LORBA encryption oracles}\label{LORBA SE}\ \\
		
		A {\it left-or-right blockwise adaptive synchronized encryption oracle} $SE(-,b)$ is an oracle applying the update of the decryption process in the following way (depending on the specified parameters).
		\begin{itemize}
			\item The function of the oracle is the same as that of a LORBA encryption oracle
			as far as the initialization of the sessions are concerned but uses $\Init_{_{\Dec}}$ to determine $\hstate(d)$.
			\item Each encryption request is of the form $(\plainvec^0,\plainvec^1,i)$ containing two plaintext-blocks
			$\plainvec^0$ and $\plainvec^1$ along with a session number $i$.
		
		 \item The oracle
			is capable of saving the history of each session, therefore, upon
			receiving such a request the oracle returns the encryption of the block $\plainvec^b$ using the history of the session $s$ along with the $\matr{IV}$ by applying the update of the decryption procedure as
	
\begin{equation}\label{eq:sssen?'}
\Enc'_{_{\theta}}(t):
\left\lbrace \begin{array}{l}
Use \ history \ of \ (\hstate(d),\matr{IV},\matr{ICV}),\\
If \ t=d\ then \\
\ \	\mathbf{c}(t)=\varepsilon'_{_{\theta}}(\xh,\matr{ICV}),\\
\ \ \hstate(t+1)=\beta'_{_{\theta}}(\xh,\matr{ICV}), \\
\ \ {\rm Output:} \ \mathbf{c}(d)\\
If\ \plainvec^b(t-d)=stop \ then \\
\ \ {\rm Output:}\ \bot\\
Else\\
\ \	\mathbf{c}(t)=\varepsilon_{_{\theta}}(\xh,\plainvec^b(t-d)),\\
\ \ \hstate(t+1)=\beta_{_{\theta}}(\xh,\mathbf{c}(t)), \\

 \ \ {\rm Output:} \ \mathbf{c}(t).
\end{array}
\right.
\end{equation}

 Note that the oracle is capable of answering different queries
related to different sessions intermittently.
	\end{itemize}

 }\end{defin}

%%%%%%%%%%%%%%%%%%%%%%%%%%%%%%%%%%%%%%%%%%%%%%%%%%%%%%%%%%%%%%%%%%%%%%
%%%%%%%%%%%%%%%%%%%%%%%%%%%%%%%%%%%%%%%%%%%%%%%%%%%%%%%%%%%%%%%%%%%%%%	

	Now we focus on the adversary model.
	
\begin{defin}{{\bf The LORBACPA$^+$ model} \label{def:exp}
		
We refer to our security model as ${\rm LORBACPA}(-,-)$ where in what follows we describe the model and its dependence on the two undetermined parameters.
\begin{itemize}
	\item A ${\rm LORBACPA}(-,-)$ adversary ${\mathcal A}$ is a nonuniform probabilistic polynomial time oracle Turing machine.
	\item The adversary uses either a LORBA encryption oracle or
	a synchronized LORBA encryption oracle or both
	of them. The second parameter denotes the oracle type as in ${\rm LORBACPA}(-,E)$ or ${\rm LORBACPA}(-,SE)$ which are referred to as {\it simple oracle} models or ${\rm LORBACPA}(-,(E \& SE))$ which is referred to as a {\it mixed oracle} model.
	\item The first parameter denotes whether the adversary can choose the initial vector
	$\matr{IV}$ or not. In particular, for an ${\rm LORBACPA}(\matr{IV},-)$ oracle the
	initial vector $\matr{IV}$ can be chosen by the adversary for each encryption session and be sent to the oracle, while for a ${\rm LORBACPA}(\$\matr{IV},-)$ oracle the initial vector of each session is chosen at random by the oracle and will be sent to the adversary in response to each one of the queries. It is instructive to
	recall that $\matr{IV}$ is the part of the initialization information that is transmitted on the communication channel.
	In what follows the security model ${\rm LORBACPA}(\matr{IV},(E \& SE))$ is referred to as LORBACPA$^+$.
	\item The adversary may query its oracle(s) about the encryption of different blocks of different messages during its computation. For each one of these messages a session must be initialized by the adversary and after that queries about consecutive blocks may be sent to the oracle. The oracle has the capability to save the history of each
	session so that the adversary may submit queries concerning the encryption of different blocks of different messages (sessions) intermittently\footnote{Sometimes this is referred to as {\it concurrent blockwise encryption} ability.}.
\end{itemize}

 The adversary simulates the following experiment $eval({\mathcal A})$.

 \begin{itemize}
	\item A key $\kappa$ is generated by $\Gen(1^k)$.
	\item A bit $b$ is chosen uniformly at random (unknown to the adversary) and the adversary is given access to its LORBA oracle(s)
	 ($E(-,b)$ and/or $SE(-,b)$) of type $b$.
	\item The adversary may send queries to its oracle concurrently.
	\item The adversary outputs a bit $b'$ as the result.
	\item The output of the experiment is defined to be $1$ (it succeeds) if $b=b'$  and $0$ otherwise.
	If it returns $1$, we say that A succeeds and otherwise it fails.
\end{itemize}

We define the {\it advantage } of an adversary
${\mathcal A}$ attacking a system ${\mathcal S}$ in the ${\rm LORBACPA}(-,-)$
model as
$$
\begin{array}{ll}
\Adv^{{\rm LORBACPA}(-,-)}_{ \mathcal{A},{\mathcal S}}(k)&=2\left|Pr(eval({\mathcal A})
=1)-\frac{1}{2}\right|\\
&=\left|[Pr(eval({\mathcal A})=1|b=1)-Pr(eval({\mathcal A})=1|b=0)]\right|\\
&=\left|[Pr_{_1}(eval({\mathcal A})=1)-Pr_{_0}(eval({\mathcal A})=1)]\right|.
\end{array}
$$
Also, the {\it insecurity} of such a system ${\mathcal S}$ in the ${\rm LORBACPA}(-,-)$ model is defined as
$$\Insec^{{\rm LORBACPA}(-,-)}_{{\mathcal S}}(k)=
\displaystyle{\max_{{\mathcal A}}}\
\Adv^{{\rm LORBACPA}(-,-)}_{ \mathcal{A},{\mathcal S}}.
$$
Clearly, such a system ${\mathcal S}$ is said to be secure in the corresponding ${\rm LORBACPA}$ model if $\Insec^{{\rm LORBACPA}(-,-)}_{{\mathcal S}}(k)$ 
is negligible compared to k e.t.
$$\Insec^{{\rm LORBACPA}(-,-)}_{{\mathcal S}}(k) < \negl(k).$$
Where $\negl(k)$ is negligible function compared to k.

}\end{defin}	

 On the other hand, for sufficiently large $k$, we have,
$$\Insec^{{\rm LORBACPA}(\${\rm IV},E)}_{{\mathcal S}}(k) \leq
\Insec^{{\rm LORBACPA}({\rm IV},E)}_{{\mathcal S}}(k)$$
and
$$\Insec^{{\rm LORBACPA}(\${\rm IV},SE)}_{{\mathcal S}}(k) \leq
\Insec^{{\rm LORBACPA}({\rm IV},SE)}_{{\mathcal S}}(k),$$
indicating that LORBACPA$^+$ is the strongest security model in this setting.

 Next, we are going to consider a fundamental type of attack based on detecting collisions, which will be our basic guideline in Section~\ref{sec:DSPLCDESCRP} for the security analysis. But, before introducing the attack let us first fix some notations and concepts to be used in what follows.
 
Note that an adversary $A$ in a ${\rm LORBACPA}(-,-)$ security model, is essentially a randomized algorithm having interactions with LORBA encryption oracles $E$ or $SE$.
In this setting we use the following notations:

\begin{itemize}
	\item $t$: is the global counter for the clock of the algorithm referring to the global real time.
	\item  $o$: refers to an oracle of type $E$ or $SE$. Note that the algorithm may initiate different sessions and ask queries intermittently and without any loss in generality we may assume that there are at most two oracles, one of type $E$ and the other of type $SE$, where they are capable of initiating and answering queries for different sessions (appropriately keeping the history of each session separately). 
	\item $i$: refers to the session number.
	\item $q(o,i,\tau,\nu)$: refers to the query asked from oracle $o$, in the $i$th session, while this query is the $\tau$th query of the $i$th session, and it is the 
	$\nu$th query asked by the algorithm (counting all queries from the beginning).
	In this setting we may write $(\plainvec^0,\plainvec^1,i)=q(o,i,\tau,\nu)$ to indicate that $(\plainvec^0,\plainvec^1,i)$ is the corresponding query in the $i$th session. 
\end{itemize}
 
 Now, let us define a {\it collision attack} as follows.
 
 \begin{defin}{
  Within a ${\rm LORBACPA}(-,-)$ security model for a self-synchronized system $S$ with delay $d$ and synchronization time $t_{_{s}}$ consider an adversary interacting with its oracle(s) $o$ and $o'$ each of which can be of type $E$ or $SE$, and assume that the adversary has initiated a number of sessions of queries. 
 In this setting,
 a {\it collision  }
 for sessions $i$ and $i'$ (not necessarily distinct) is the event
 $Coll((o,i,\tau),(o',i',\tau'),b)$
 for which the adversary have the following collision for the ciphertexts,
 $$\varepsilon_{_{\theta}}(\state(\tau+d\delta_{o}),\plainvec^b_{i,o}(\tau))=\varepsilon_{_{\theta}}(\state'(\tau'),\tilde{\plainvec}^b_{i',o'}(\tau'))$$
 and 
 $$\varepsilon_{_{\theta}}(\state(\tau+d\delta_{o}),\plainvec^{\bar{b}}_{i,o}(\tau)) \not =\varepsilon_{_{\theta}}(\state'(\tau'),\tilde{\plainvec}^{\bar{b}}_{i',o'}(\tau'))$$
 where $\delta_{_{SE}}=1$, $\delta_{_E}=0$ and we have 
 $$(\plainvec^0_{i,o},\plainvec^1_{i,o},i)=q(o,i,\tau,\nu) \quad {\rm and }  \quad
 (\tilde{\plainvec}^0_{i',o'},\tilde{\plainvec}^1_{i',o'},i)=q(o',i',\tau',\nu').$$ 
 
 A ${\rm LORBACPA}(-,-)$ {\it collision attack} is a ${\rm LORBACPA}(-,-)$ adversary for which the probability  of detecting a collision is non-negligible.
}\end{defin}

Clearly, the main objective of a collision attack is to provide an algorithm (whose output depends on the queries) that can maximize the probability of a collision. For this, a straight forward approach is to provide a sequence of queries in different sessions such that it leads to a collision. To see this, note that the attack provided 
in Section~\ref{sec:Sm} for the DCBC mode is essentially a collision attack using only one oracle of type $SE$.

We should emphasize that the straight forward scenario mentioned above is not the only possible setup for a collision attack where one may think of many different approaches 
to detect collisions. To see one more scenario, consider an adversary who uses one oracle of type $E$ and one other oracle of type $SE$, while the adversary tries to detect $t_{_{s}}\label{key}$ consecutive ciphertexts such that 
$$\forall\ 0 \leq j \leq t_{_{s}}-1 \quad \ciphervec_{E}^{b}((\tau-t_{_{s}})+j)=\ciphervec_{SE}^{b}((\tau'+d-t_{_{s}})+j),$$
in which the flag $b$ refers to the oracle action type (i.e. left or right).
Note that  as a consequence of synchronization we have $\state(\tau) = {\hstate}(\tau'+d)$.

Then choosing $\plainvec_{i,E}^0(\tau)={\plainvec}_{i',SE}^0(\tau')$ for $b=0$ we have,

{\scriptsize	
	\begin{equation}\label{lab:DSPLC2}
	\begin{array}{ll}
	\ciphervec_{E}^{b}(\tau+d)-{\ciphervec}_{SE}^{b}(\tau'+d)=\varepsilon_{_{\theta}}(\state(\tau),\plainvec_{i,E}^0(\tau))-\varepsilon_{_{\theta}}({\hstate}(\tau'+d),{\plainvec}_{i',SE}^0(\tau'))=0.
	\end{array}
	\end{equation}
}	
On the other hand, choosing a random $\plainvec_{i,E}^1(\tau)=\plainvec_{_{!}} \not = 0$ 
for $b = 1$  and setting ${\plainvec}_{i',SE}^1(\tau')=0$ we have,
{\scriptsize
	\begin{equation}\label{lab:DSPLC22}
	\begin{array}{ll}
	\ciphervec_{E}^{b}(\tau+d)-{\ciphervec}_{SE}^{b}(\tau'+d)&=\varepsilon_{_{\theta}}(\state(\tau),\plainvec_{i,E}^1(\tau))-\varepsilon_{_{\theta}}({\hstate}(\tau'+d),{\plainvec}_{i',SE}^1(\tau'))\\
	&=\varepsilon_{_{\theta}}(\state(\tau),\plainvec_{_{!}})-\varepsilon_{_{\theta}}(\state(\tau),{\bf 0}).
	\end{array}
	\end{equation}
}	
Then effectiveness of the algorithm follows from the fact that $\varepsilon_{_{\theta}}$ is collision resistant.

\begin{cor}\label{cor:main}
A self-synchronized stream cipher $S$ for which there exists $\tau_{_{0}}$ and a function $f$ such that
$\state(\tau_{_{0}}) = f(\matr{IV})$ (resp. $\hstate(\tau_{_{0}}) = f(\matr{IV})$), is not ${\rm LORBACPA}(\matr{IV},E)$ (resp. ${\rm LORBACPA}(\matr{IV},SE)$) secure.
\end{cor}
\begin{proof}{ For a collision attack, the adversary initializes two sessions $i$ and $i'$ with the same initial values $\matr{IV}$, respectively. 
		Then it chooses $\plainvec_{i,o}^0(\tau_{_{0}})=\plainvec_{i,o}^1(\tau_{_{0}})=\tilde{\plainvec}_{i',o}^0(\tau_{_{0}})=\plainvec_{_{!}} \not = 0$ uniformly at random and set $\tilde{\plainvec}_{i',o}^1(\tau_{_{0}})=0$. Since $\state(\tau_{_{0}})={\state}'(\tau_{_{0}})=f(\matr{IV})$ (resp. $\hstate(\tau_{_{0}})={\hstate}'(\tau_{_{0}})=f(\matr{IV})$), this gives rise to a collision for the queries
		$$(\plainvec_{_{!}} ,\plainvec_{_{!}} ,i)=q(o,i,\tau_{_{0}},\nu) \quad {\rm and }  \quad
		(\plainvec_{_{!}} ,0,i')=q(o,i',\tau_{_{0}},\nu').$$ 
}\end{proof}

 An important consequence of this corollary is the fact that the state vectors of a secure
${\rm LORBACPA}(\matr{IV},-)$ self-synchronized stream cipher must possess a pseudorandom nature and should not be a deterministic function of the initial vector $\matr{IV}$.

Before we proceed any further, let us consider some well-known self-synchronized block-cipher modes and their security in our proposed models.

 In \cite{JMV02} A.~Joux~{\it et.al.} have shown that the $\matr{CBC}$ encryption mode cannot be IND secure in the blockwise adversarial model,
which implies that, in our language, the $\matr{CBC}$ mode is not
${\rm LORBACPA}(\$\matr{IV},E)$ secure. Subsequently, P.~Fouque~{\it et.al.} \cite{FMP03} using an output delay procedure, introduced the {\it delayed } $\matr{CBC}$ mode, $\matr{DCBC}$ (see Example~\ref{exm:DCBC}), and proved that it is secure in the
blockwise model (assuming the security of the underlying block cipher).
This implies that the $\matr{DCBC}$ mode is
${\rm LORBACPA}(\$\matr{IV},E)$ secure under the same assumptions
(note that our ${\rm LORBACPA}(\$\matr{IV},E)$ security is equivalent to
the ${\rm LORC-BCPA}$ security in \cite{FMP03}).
On the other hand, the $\matr{DCBC}$ mode is not
${\rm LORBACPA}(\matr{IV},-)$ secure by Corollary~\ref{cor:main}.
Clearly, the $\matr{DCBC}$ mode is also not secure in the
${\rm LORBACPA}(\$\matr{IV},SE)$ setting by the collision attack presented at the beginning of this section.

 For the $\matr{CFB}$ mode, by the results of \cite{FMP03}, we know that the scheme is provably secure in the blockwise model, assuming the security of the underlying block cipher (or function) which is equivalent to our
${\rm LORBACPA}(\$\matr{IV},E)$ security. It is instructive to note that the
update and output functions of an $SE(\$\matr{IV},b)$ oracle for the $\matr{CFB}$ mode
is identical to those of the oracle $E(\$\matr{IV},b)$. This, in particular, proves that
the $\matr{CFB}$ mode is also ${\rm LORBACPA}(\$\matr{IV},SE)$ secure.
On the other hand, the $\matr{CFB}$ mode is not
${\rm LORBACPA}(\matr{IV},-)$ secure by Corollary~\ref{cor:main}.

 \subsection{A modification}
 
 This section is going to serve as an appetizer before we focus on our final proposal 
 in the next section. 
 Based on our results of the previous section we understand that we have to choose a 
 pseudorandom state update procedure in order to prevent insecurity in LORBACPA$^+$ 
 model.  Our major objective in this section is to analyze the performance of a modification on well-known encryption modes based on choosing this approach and adding iterations of a nilpotent linear function for finite-time self-synchronization to the scheme. 
 We will see that although this modification will not give rise to better encryption modes but the analysis will pave the way to introduce our proposed scheme in the next section in which we will have to prevent linearity and we also force the iteration function to depend on the secret key.
 
 \subsubsection{A modified $\matr{DCBC}$ mode}\label{sec:MDCBC}
 
 Let us introduce the modified $\matr{DCBC}$ mode as follows.

\begin{defin}{The $\matr{MDCBC}$ mode.
\begin{itemize}
	\item[-] ${f}$: is a known linear function with a natural number 	$n_{_{0}} > 1$ such that
	 $$\forall\ x\, \ \ {f}^{n_{_{0}}}(x) \isdef {f}({f}(...({f}(x))))=0.$$

 	\item[-]
	{\scriptsize
	\begin{equation}\label{esssc2}
	\begin{array}{ll}
	\Enc_{_{\theta}}(t):
	\left\lbrace \begin{array}{l}
			{\rm Initialize }\ \ \matr{ IV},\\
	\state(0) \overset{\$}{\leftarrow} \Field_{_{2}}^{n},\\
	If \ t=0 \ then\\		
	\ \ \mathbf{c}(1) =\matr{IV},\\
	 \ \ \state(1) =\matr{IV}\oplus {f}(\state(0)),\\
	 \ \ {\rm Output:} \ \bot\\
	Else \\
	\ \ \mathbf{c}(t+1) = E_{\key}(\state(t) \oplus \plainvec(t)),\\
	\ \ \state(t+1) = E_{\key}(\state(t) \oplus \plainvec(t)) \oplus {f}(\state(t)),\\
	 \ \ {\rm Output:} \ \mathbf{c}(t).\\
	\end{array}
	\right.
	&
	\Dec_{_{\theta}}(t):
	\left\lbrace \begin{array}{l}
	Receive \ \matr{ IV},\\
	\hstate(1) \overset{\$}{\leftarrow} \Field_{_{2}}^{n},\\
	If \ t=1 \ then \\
	 \ \ \hstate(2)=\matr{IV}\oplus {f}(\hstate(1)),\\
\ \ \hp(t)=\ {\rm Ack},\\
\ \ {\rm Output:} \ \hp(t)\\
	Else\\
	\ \ \hstate(t+1) = \mathbf{c}(t)\oplus {f}(\hstate(t)),\\
	\ \ \hp(t) = D_{\key}(\mathbf{c}(t)) \oplus \hstate(t),\\
\ \ {\rm Output:} \ \hp(t).\\
	\end{array}
	\right.
	\end{array}
	\end{equation}
}

 \end{itemize}
Note that for the $\matr{MDCBC}$ mode we have $t_{_{c}}=1$ and $d=1$. In what follows we show that also $t_{_{s}}= n_{_{0}}$.

 }\end{defin}

 \begin{lem} The $\matr{MDCBC}$ mode is {\it finite-time self-synchronized} with $t_{_{s}}= n_{_{0}}$, i.e.
$$ \forall t \geq t_{_{s}}= n_{_{0}}, \ \ \hstate(t+1)=\state(t).$$
\end{lem}
\begin{proof}{		\begin{equation}\label{lab:MDCBC}
		\begin{array}{ll}
		\ke&\isdef \hstate(t+2)-\state(t+1)\\
	&\ = \mathbf{c}(t+1)\oplus {f}(\hstate(t+1))-(E_{\key}(\state(t) \oplus \plainvec(t))\oplus {f}(\state(t)))\\
		&\ ={f}(\hstate(t+1))-{f}(\state(t))
		 ={f}(\hstate(t+1)-\state(t))={f}(\ek).
		\end{array}
		\end{equation}
		Hence, according to the definition of the $\matr{MDCBC}$ system, the error is equal to zero for $t \geq t_{_{s}}= n_{_{0}}$.
}\end{proof}
	
The following lemma shows that our modification is not effective in the presence of an $SE$ oracle.
\begin{lem}\label{lem:MDCBC2}
The $\matr{MDCBC}$ mode is not ${\rm LORBACPA}(\$\matr{IV},SE)$ secure.
\end{lem}
\begin{proof}{
Recall that an $SE(\$\matr{IV},b)$ oracle operates as follows,
\begin{equation}\label{eq:sssen'}
\Enc'_{_{\theta}}(t):
\left\lbrace \begin{array}{l}
	Initialize \ \matr{ IV},\\
	\hstate(1) \overset{\$}{\leftarrow} \Field_{_{2}}^{n},\\
	If \ t=1 \ then \\
\ \ \hstate(2)={f}(\hstate(1))\oplus\matr{IV},\\
\ \ \mathbf{c}(1)= \matr{IV},\\
\ \ {\rm Output:} \ \mathbf{c}(1)\\
If\ \plainvec(t-1)=stop \ then \\
\ \ {\rm Output:}\ \bot\\
Else\\
\ \	\mathbf{c}(t)=E_{\key}(\hstate(t) \oplus \plainvec(t-1)),\\
\ \ \hstate(t+1)=\mathbf{c}(t)\oplus {f}(\hstate(t)), \\

 \ \ {\rm Output:} \ \mathbf{c}(t).
\end{array}
\right.
\end{equation}

We introduce a collision attack whose success probability is equal to one.
For this, the adversary initializes two sessions $i$ and $i'$ with two random initial values $\matr{IV}$
and $\tilde{\matr{IV}}$, respectively. Then it chooses at random a sequence $\{(\plainvec_{i,SE}^0(\tau),\plainvec_{i,SE}^1(\tau))\} , 0<\tau<t_{_{s}} $ with 
$\plainvec_{i,SE}^0(t_{_{s}})=\plainvec_{i,SE}^1(t_{_{s}})=\plainvec_{_{!}} \not = 0$ uniformly at random and receives $\{\ciphervec^{b}(1),\cdots ,\ciphervec^{b}(t_{_{s}}) \}$. Then it chooses two random sequences $\{(\tilde{\plainvec}_{i',SE}^0(\tau),\tilde{\plainvec}_{i',SE}^1(\tau))\} , 0<\tau<t_{_{s}}  $ and receives $\{\tilde{\ciphervec}^{b}(1),\cdots ,\tilde{\ciphervec}^{b}(t_{_{s}}) \}$. Then it sets $\tilde{\plainvec}_{i',SE}^1(t_{_{s}})=0$ and $\tilde{\plainvec}_{i,SE}^0(t_{_{s}})=\plainvec_{_{*}}$, with

 $$\plainvec_{_{*}}=\plainvec_{_{!}}\oplus\ciphervec^{b}(t_{_{s}})\oplus\cdots\oplus {f}^{t_{_{s}}-1}({\matr{IV}})\oplus\tilde{\ciphervec}^{b}(t_{_{s}})\oplus\cdots\oplus {f}^{t_{_{s}}-1}(\tilde{\matr{IV}})$$ for session $i'$. 
Since we have 
$$ \hstate(t_{_{s}}+1)=\mathbf{c}^{b}(t_{_{s}})\oplus {f}(\ciphervec^{b}(t_{_{s}}-1))\oplus\cdots\oplus {f}^{(t_{_{s}}-1)}({\matr{IV}})$$
and 
$$ {\hstate'}(t_{_{s}}+1)=\tilde{\mathbf{c}}^{b}(t_{_{s}})\oplus {f}(\tilde{\ciphervec}^{b}(t_{_{s}}-1))\oplus\cdots\oplus {f}^{(t_{_{s}}-1)}(\tilde{\matr{IV}}),$$
 this gives rise to a  collision for the queries
$$(\plainvec_{_{!}} ,\plainvec_{_{!}} ,i)=q(SE,i,t_{_{s}}+1,\nu) \quad {\rm and }  \quad
(\plainvec_{_{*}} ,0,i')=q(SE,i',t_{_{s}}+1,\nu').$$

 }\end{proof}

 \subsubsection{A modified $\matr{CFB}$ mode}\label{sec:MCFB}

Now let us concentrate on a modified version of the $\matr{CFB}$ mode as follows,

 \begin{defin}{The $\matr{MCFB}$ mode.

 \begin{itemize}
	\item[-] ${f}$: is a known linear function such that
	there exists a natural number 	$n_{_{0}} > 1$ such that
	$$\forall\ x\, \ \ {f}^{n_{_{0}}}(x) \isdef {f}({f}(...({f}(x))))=0.$$
	
	\item[-]
	{\scriptsize
	\begin{equation}\label{esssc3}
	\begin{array}{ll}
	\Enc_{_{\theta}}(t):
	\left\lbrace \begin{array}{l}
	Initialize \ \matr{ IV},\\
\state(0) \overset{\$}{\leftarrow} \Field_{_{2}}^{n},\\
	If \ t=0 \ then \\
	\ \ \state(1) = E_{\key}(\matr{IV} )\oplus {f}(\state(0)),\\
	\ \ {\rm Output:} \ \bot\\
	Else\\
	\ \ \kx = E_{\key}(\mathbf{c}(t) )\oplus {f}(\state(t)),\\
	\ \ \mathbf{c}(t) = \uk\oplus \xk,\\
	{\rm Output:} \ \mathbf{c}(t).
	\end{array}
	\right.
	&
	\Dec_{_{\theta}}(t):
	\left\lbrace \begin{array}{l}
	Receive \ \matr{ IV},\\
	\hstate(0) \overset{\$}{\leftarrow} \Field_{_{2}}^{n},\\
	If \ t=0 \ then \\
	\ \ \hstate(1) = E_{\key}(\matr{IV} )\oplus {f}(\hstate(0)),\\
	\ \ {\rm Output:} \ Ack\\
	Else\\
	\ \ \hx = E_{\key}(\mathbf{c}(t))\oplus {f}(\hstate(t)),\\
	\ \ \uh=\xh\oplus\mathbf{c}(t),\\
	\ \ {\rm Output:} \ \uh.
	\end{array}
	\right.
	\end{array}
	\end{equation}
}
\end{itemize}
Note that for the $\matr{MCFB}$ mode we have $t_{_{c}}=1$ and $d=0$. In what follows we show that also $t_{_{s}}= n_{_{0}}$.
}\end{defin}

 \begin{lem} The $\matr{MCFB}$ mode is {\it finite-time self-synchronized} with $t_{_{s}}= n_{_{0}}$, i.e.
	$$ \forall t \geq t_{_{s}}, \ \ \hstate(t)=\state(t).$$
\end{lem}
\begin{proof}{		\begin{equation}\label{lab:MCFB }
		\begin{array}{ll}
		\ke& \isdef \hstate(t+1)-\state(t+1)\\
		& \ = E_{\key}(\mathbf{c}(t))+{f}(\hstate(t))-E_{\key}(\mathbf{c}(t))-{f}(\state(t))\\
		& \ ={f}(\hstate(t))-{f}(\state(t))
		={f}(\hstate(t)-\state(t))={f}(\ek).
		\end{array}
		\end{equation}
		Hence, according to the definition of the $\matr{MCFB}$ system, the error is equal to zero for $t \geq t_{_{s}}= n_{_{0}}$.
}\end{proof}
%######################################
\begin{lem}\label{lem:MCFB2}
	The $\matr{MCFB}$ mode is not ${\rm LORBACPA}(\matr{IV},E)$ secure. 	
\end{lem}
\begin{proof}{
		 The $\matr{MCFB}$ mode is $\matr{CFB}$ mode with a linear transformation    $\matr{MCFB}$ mode is a simple reduction of $\matr{MCFB}$ mode. For any $t>0$, with Linear transformation $f_\key^{t_{_{s}}-1}(.)$ over encryption update and output functions $\matr{MCFB}$ mode, we have
			\begin{equation}\label{esssc4}
		\begin{array}{ll}
		\Enc"_{_{\theta}}(t):
		\left\lbrace \begin{array}{l}
		Initialize \ \matr{ IV},\\
		\state(0) \overset{\$}{\leftarrow} \Field_{_{2}}^{n},\\
		If \ t=0 \ then \\
		\ \ \state(1) =E_{\key}(\matr{IV} )\oplus {f}(\state(0)),\\
		\ \ {\rm Output:} \ \bot\\
		Else\\
		\ \ \kx = E_{\key}(\mathbf{c}(t) )\oplus {f}(\state(t)),\\
		\ \ \mathbf{c}'(t) ={f}^{t_{_{s}}-1}(\mathbf{c}(t))={f}^{t_{_{s}}-1}(\uk)\oplus {f}^{t_{_{s}}-1}(\xk),\\
		{\rm Output:} \ \mathbf{c}(t).
		\end{array}
		\right.
		&
		\end{array}
		\end{equation}

		The following collision attack finds a collision with probability $1$. For this,
		the adversary initializes two sessions $i$ and $i'$ with the same initial values $\matr{IV}$. Then it chooses $\plainvec_{i,E}^0(1)=\plainvec_{i,E}^1(1)=\plainvec_{_{!}} \not = 0$ uniformly at random. On the other hand, it sets  $\tilde{\plainvec}_{i',E}^0(1)=\plainvec_{_{!}}, \tilde{\plainvec}_{i',E}^1(1)=0$.
		
		Since $\mathbf{c}'(t)={f}^{t_{_{s}}-1}(\state(t)\oplus \uk)$, $\state(1)= E_{\key}(\matr{IV} )\oplus {f}(\state(0))$ and $\tilde{\state}(1)=E_{\key}(\matr{IV} )\oplus {f}(\tilde{\state}(0))$, this gives rise to a collision for the queries
	$$(\plainvec_{_{!}} ,\plainvec_{_{!}} ,i)=q(E,i,1,\nu) \quad {\rm and }  \quad
	(\plainvec_{_{!}} ,0,i')=q(E,i',1,\nu').$$ 			
}\end{proof}
%######################

This shows that $\matr{MCFB}$ mode is not
${\rm LORBACPA}(\matr{IV},-)$ secure.
The following table summarizes our results so far.

%	\caption{ The summarizes our results for $\matr{S}_{\sigma}^4(\mathfrak{P})$,$\matr{CBC}$ ,$\matr{DCBC}$ ,$\matr{CFB}$ cryptosystem.}\label{Tab:M}
\begin{center}\label{tbl:t1?}
\begin {table}[H]
\caption {Security results for $\matr{CBC}$, $\matr{CFB}$, $\matr{DCBC}$, and $\matr{S}_{\sigma}^4(\mathfrak{P})$  cryptosystems.} \label{tab:title} 
	
\begin{tabular}{|c|c|c|c|c|c|c|c|}
	
\hline {\tiny System$/{\rm LORBACPA}$}
&{\footnotesize $(\$\matr{IV},E)$} & {\footnotesize $(\matr{IV},E)$} & {\footnotesize $(\$\matr{IV},SE)$} & {\footnotesize $(\matr{IV},SE)$} & {\footnotesize $ (\$\matr{IV},(E\&SE))$}& {\footnotesize $^+$}& Ref. \\
\hline $\matr{CBC}$&$\times$&$\times$&$\times$&$\times$&$\times$&$\times$&\cite{FMP03}\\
\hline $\matr{DCBC}$&$\surd$&$\times$&$\times$&$\times$&$\times$&$\times$&
\cite{FMP03}, Example~\ref{exm:DCBC}, Cor.~\ref{cor:main} \\
\hline$\matr{CFB}$&$\surd$&$\times$&$\surd$&$\times$&$\surd$&$\times$&\cite{FMP03}, Cor.~\ref{cor:main}\\
\hline $\matr{S}_{\sigma}^4$&$\surd$&$\surd$&$\surd$&$\surd$&$\surd$&$\surd$& Sec.~\ref{sec:cryptoanalysis}\\
\hline
\end{tabular}

\end {table}
\end{center}

\section{The $\matr{S}_{\sigma}^4(\mathfrak{P})$ cryptosystem }\label{sec:DSPLCDESCRP}

 In this section we define our proposed\footnote{The acronym stands for {\it Switching Self-Synchronized Stream-cipher} } cryptosystem
$\matr{S}_{\sigma}^4(\Gen,\transmitter_{\kappa},\receptor_{\kappa},\mathfrak{P})$ which is based on
basic ideas comming from the design of cryptographic modes of operations and the contributions of G.~Mill\'{e}rioux et.al. 
\cite{PM13} in design and analysis of switching cryptosystems.

 \subsection{System description}

 Here we go through the details of our system's description.

 \begin{itemize}
	\item{{\bf Definition of parameters}
		\begin{itemize}
			\item{The integer $n$ is system's dimension which also determines the length of each block of plaintext or ciphertext.}
	 \item The integer $q$ stands for the size of the finite field $\Field_{_{q}}$.
	 \item The security parameter is $k=n\log q$.
	 \item The key generator $\Gen$: is a probabilistic algorithm that on input $1^k$
	outputs the secret key $\key$. 

 	 \item The family $\mathfrak{P}$: is a family of pseudorandom permutations  on the elements of $GF(q)$ as $ \pi_{_{\key}}:\Field_{_{q}} \outputs \Field_{_{q}} $ indexed by the (secret)  key string $\key$. 
 	 
 	 In this setting the one-to-one function $\BigPi_{_{\key}}:\Field^{^{n}}_{_{q}} \outputs \Field^{^{n}}_{_{q}}$(resp. $ \BigPi_{_{\key}}^{(-1)}:\Field^{^{n}}_{_{q}} \outputs \Field^{^{n}}_{_{q}}$) is defined by the action of the (secret) random permutation $ \pi_{_{\key}} $ (resp. $ \pi^{(-1)}_{_{\key}} $) on each entry of an $n$-vector in $\Field^{^{n}}_{_{q}}$.

	\item{The switching function $\sigma(t):\mathbb{N}\rightarrow [ \ell ]$ (see \cite{PM13}): 
		The switching function must depend on the output of the system, while the
		motivation of such a dependence lies in that the switching rule must also be self-synchronizing. Thus, it must depend on	shared variables and so on the output or a finite sequence of delayed outputs.  
		
%	is the switching function implemented through the output of a linear feedback shift register with a secret initializing vector $LFSR_{0}$.
		
%	The symbol $\matr{S}^4(\mathfrak{P})$ is reserved to refer to the case when the system does not use any switching function.

    }

	\item{Matrices: The matrix $\Wmat$ and for $j \in [ \ell ]$, the matrices  $\Lmat_{j}, \Fmat_{j}$ are secret invertible matrices
				in $\mathcal{M}^{n \times n}$. The set $\{ {\Q}_{j} \ |\ j \in [ \ell ]\}$ is a public nilpotent semigroup of matrices of index $n_{_{0}}$ (see e.g. \cite{PM13}).
				
				On the other hand, for any $j \in [ \ell ]$, the matrices $\Emat_{j},\Bmat_{j}$ in $\mathcal{M}^{n \times n}$ are public invertible matrices,
				where, for any $j \in [ \ell ]$, we define
				$\Amat_{j} \isdef \Emat_{j}\Fmat_{j}^{-1}\Bmat_{j}$, $\Rmat_{j} \isdef \Emat_{j}\Fmat_{j}^{-1}$ and $\Dmat_{j} \isdef \Emat_{j}\Fmat_{j}^{-1}\Lmat_{j}-\Q_{j},$.
			}	
				\item{The matrix $\memM$ is an
				upper triangular public matrix  with zeros on the diagonal in $\mathcal{M}^{m_{_{0}} \times m_{_{0}}}$.
			}				
			\item{At any time $t \in \mathbb{N}$, $\plain(t) \in \Field_{_{q}}$ and $\cipher(t) \in \Field_{_{q}}$ are
				the plaintext and the ciphertext at time $t$. 				
				$$\plainvec(t)=(\plain((t-1)n+1), \plain((t-1)n+2), \cdots, \plain(tn)) ^{T}$$
				and
				$$ \ciphervec(t)=(\cipher((t-1)n+1), \cipher((t-1)n+2), \cdots, \cipher(tn)) ^{T}$$
				are the $t\geq 1$'th blocks of plaintext and ciphertext, respectively.
         	We assume that each block of data is of length $n$
whose symbols are numbered from 1 to $n$, and that the end of encryption is indicated by sending a predefined block $\plainvec(t) = stop$. Also, we assume that if the decryption algorithm does not have
to output a block, it sends, as an acknowledgment, a predefined block $Ack$.
			}
		\end{itemize}
		
	}
	
	\item{{\bf Encryption procedure ($\Enc_{_{\key}}$)}\\

		\begin{figure}[t]
		\centering
		\begin{tikzpicture} [
		auto,
		decision/.style = { diamond, draw=black, thick,
			text width=5em, text badly centered,
			inner sep=1pt },
		block/.style = { rectangle, draw=black, thick,
			text width=14em, text centered,
			minimum height=2em },
		block1/.style = { rectangle, draw=black, thick,
			text width=7em, text centered,
			minimum height=1em },
		line/.style = { draw, thick, ->, shorten >=3pt },
		cl/.style={ultre thick, draw, ellipse, node distance=3cm, minimum height=2em},
		]
		% Define nodes in a matrix
		\matrix [column sep=20mm, row sep=5mm]
		{ &\node [block1] (x1) {\scriptsize{$\begin{array}{l}
					Gen(1^k)\rightarrow \key\\
					\end{array}$}};&\node [block1] (x11) {\scriptsize{$\begin{array}{l}
					Gen(1^k)\rightarrow \key\\
					\end{array}$}};&\\
			\node [text centered] (x4) { };&\node [block] (x2) { \scriptsize{$
					\begin{array}{l}
					(\state(0),\matr{ IV})\ \overset{\$}{\leftarrow} \Field_{_{q}}^{n},\\
					{\rm Initial:} \{
					\pi_{_{\key}},\sigma,
					\{\Lmat_{j}\}_{j=0}^{\ell},\{\Fmat_{j}\}_{j=0}^{\ell}
				\}\\
					\rm {Public}: \{\Bmat_{j}\}_{j=0}^{\ell}, \{\Emat_{j}\}_{j=0}^{\ell}, \memM, \\\ \ \ \ \ \ \ \ \ \{{\Q}_{j}\}_{_{j=0}}^{^{\ell}}.\\ \rm{Compute}: \{\Amat_{j}\}_{j=0}^{\ell}, \{\Dmat_{j}\}_{j=0}^{\ell} .\\ \\
					\transmitter_{\kappa}: \left\lbrace
					\begin{array}
					{l}
					{\rm Input:}\lbrace\plainvec(t)\rbrace \\
					\kernel\:\transmitter(\plainvec(t))\\
					{\rm Output:} (\{\ciphervec(t)\}, \matr{ IV})\\
					\end{array}
					\right. \\
					\end{array}$}};& \node [block] (y2) {\scriptsize{$
					\begin{array}{l}
					(\hstate(1),\matr{ IV})\ \overset{\$}{\leftarrow} \Field_{_{q}}^{n},\\
					{\rm Initial:} \{
					\pi_{_{\key}},\sigma,
					\{\Lmat_{j}\}_{j=0}^{\ell},\{\Fmat_{j}\}_{j=0}^{\ell}\}\\
					\rm{Public}:\{\Bmat_{j}\}_{j=0}^{\ell}, \{\Emat_{j}\}_{j=0}^{\ell}, \memM, \\\ \ \ \ \ \ \ \ \ \{{\Q}_{j}\}_{_{j=0}}^{^{\ell}}.\\\rm{Compute}: \{\Amat_{j}\}_{j=0}^{\ell}, \{\Dmat_{j}\}_{j=0}^{\ell}.\\ \\
					\receptor_{\kappa}: \left\lbrace
					\begin{array}
					{l}
					{\rm Input:}(\{\ciphervec(t)\}, \matr{IV}) \\
					
					\kernel\receptor (\ciphervec(t))\\
					{\rm Output:} \{\widehat{{\plainvec}}(t)\}\\
					\end{array}
					\right. \\
					\end{array}$}};&\node [text centered] (x5) { };&\\ };
		%
		% legend for subprocedures
		\begin{scope} [every path/.style=line]
		\path (x11) -- node [] {} (y2);
		
		\path (x1) -- node [] {} (x2);
		\path (x4) -- node [] {$\{\plainvec(t)\}$} (x2);
		\path (y2) -- node [] {$\{\widehat{\plainvec}(t)\}$} (x5);
		\path (x2) -- node [] { $\{\ciphervec(t)\}, \matr{ IV}$} (y2) ;
		\end{scope}
		\end{tikzpicture}

		\caption{ The $\matr{S}_{\sigma}^4(\mathfrak{P})$ cryptosystem.
		}\label{fig:DSPLC1}
	\end{figure}
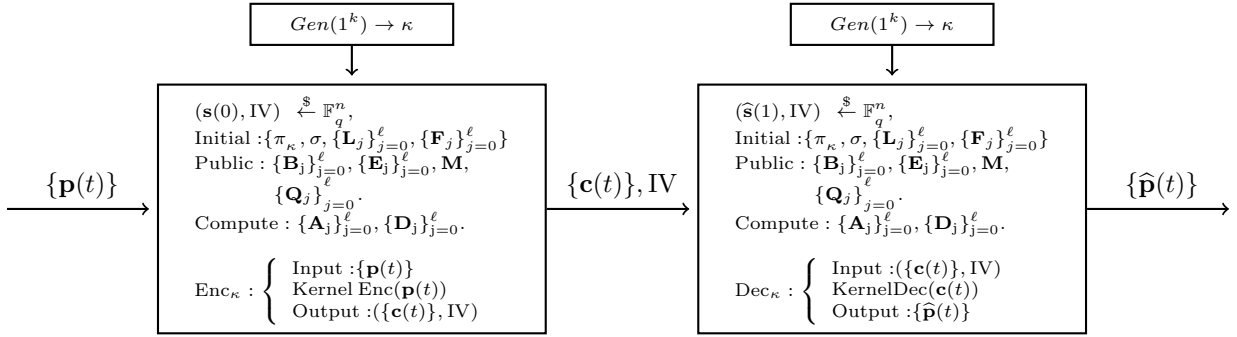
		\begin{itemize}
			\item The transmitter chooses a vector
			$\state(0)=(s(0)_{1} \cdots, s(0)_{n}) ^{T}$
			at random in $\Field_{_{q}}^{n}$.
			\item The transmitter chooses two random vectors $\memvec(0),\ciphervec(0) \in \Field_{_{q}}^n$ and forms the initial-value vector
			$\matr{ IV} \isdef (\memvec(0),\ciphervec(0))$ with
			$$\ciphervec(0)=(\cipher(1-n), \cdots, \cipher(0)) ^{T},$$
			$$\memvec(0)=(m(0)_{1} \cdots, m(0)_{m_{_{0}}}) ^{T}$$
			and transmits this vector over the public channel.
		
			\item The general description of the encryption procedure is as follows (see Figures~\ref{fig:DSPLC1} and \ref{fig:DSPLC2}).
			{\scriptsize		
				\begin{flushleft}
					$
					\transmitter_{\kappa}: \left\lbrace
					\begin{array}
					{ll}
					\rm Input:&\lbrace\plainvec(t), t\geq 1\rbrace \\
					{\rm Initial:}&\pi_{_{\key}},\sigma(t), \state(0),\matr{IV}, \lbrace \Bmat_{j}\rbrace, \lbrace \Emat_{j}\rbrace, \lbrace \Fmat_{j}\rbrace, \lbrace \Lmat_{j}\rbrace, \Wmat, \memM\\
					\kernel\:\transmitter(\plainvec(t)), t\geq 1\\
					
					{\rm Output:} &(\{\ciphervec(t), t\geq 1\}, \matr{ IV}).\\
					\end{array}
					\right.
					$
				\end{flushleft}
				%%%%%%%%%%%%%%%%%%%%%%%%%%%%%%%%%%%%%%%%%
				\[
				\kernel\:\transmitter(\plainvec(t)):\left\lbrace
				\begin{array}
				{ll}
				\matr{IV}=(\memvec(0),\mathbf{c}(0)),\ \ \state(0) \overset{\$}{\leftarrow} \Field_{_{q}}^{n},\\
				If \ t=0 \ then&\\	
				\ \ \left\lbrace
				\begin{array}	 {l}	
				\state(1) =\state(0), \\
				\memvec(1)=\memvec(0),\\
				\mathbf{c}(1) = \ciphervec(0),\\
				\end{array}
				\right.\\
				\ \ {\rm Output:} \bot.\\
				If \ t=1 \ then \\
				\ \ \left\lbrace
				\begin{array} {l}
				\mathrm{update\: \transmitter},\\
				\ciphervec(2)=\encryptionfun_{_{\key}}(\keystreamvec(1), \plainvec(1)),\\
				\end{array}
				\right.\\
				\ \ {\rm Output:} \ciphervec(1)\\
				
				If\ \plainvec(t)=stop \ then \\
				\ \ {\rm Output:} \ciphervec(t)\\
				Else \\
				\ \ \left\lbrace
				\begin{array}
				{l}\mathrm{update\: \transmitter},\\
				\ciphervec(t+1)=\encryptionfun_{_{\key}}(\keystreamvec(t), \plainvec(t)),\\
				\end{array}
				\right.\\
				\ \ \ \ {\rm Output:} \ciphervec(t).\\
				\end{array}
				\right.
				\]
				%%%%%%%%%%%%%%%%%%%%%%%%%%%%%%%%%%%%%%%%%%%%%
							
			}		
			where the details of the functions in update procedure are as follows,
			{\scriptsize
				\begin{equation}\label{lab:DSPLC}
					\begin{array}
						{llll}
						\mathrm{update\: \transmitter}:\left\lbrace
						\begin{array}{llll}
										\rm Initial &: j=\sigma(t),\\
							\state(t+1)&=\Wmat\memvec(t)+\Dmat_{j}\state(t)+ \Amat_{j}\BigPi_{_{\key}}(\state(t))+\Emat_{j}\BigPi_{_{\key}}(\plainvec(t)),\\
							\memvec(t+1)&=\memM\memvec(t)+\ciphervec(t),\\
							\keystreamvec(t)&=\Lmat_{j}\state(t)+ \Bmat_{j}\BigPi_{_{\key}}(\state(t)),\\
						\end{array}
						\right.
					\end{array}
				\end{equation}
				\begin{equation}\label{lab:DSPLC1}
				{\rm Output} \: \transmitter:\ciphervec(t+1)=\keystreamvec(t)+\Fmat_{j}\BigPi_{_{\key}}(\plainvec(t)).\\
				\end{equation}
			}	
		\end{itemize}
		
		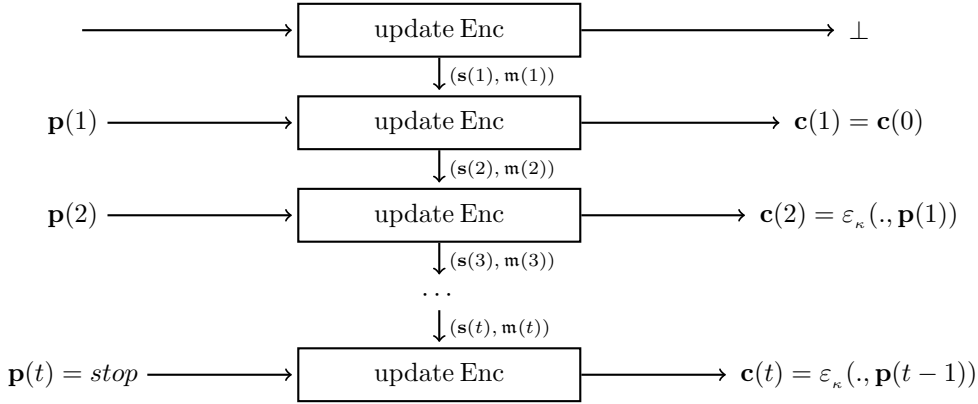
\begin{figure}[t]
			\centering
			\begin{tikzpicture} [
			auto,
			decision/.style = { diamond, draw=black, thick,
				text width=5em, text badly centered,
				inner sep=1pt },
			block/.style = { rectangle, draw=black, thick,
				text width=10em, text centered,
				minimum height=2em },
			line/.style = { draw, thick, ->, shorten >=2pt },
			]
			
			% Define nodes in a matrix
			\matrix [column sep=20mm, row sep=5mm]
			{ \node [text centered] (x1) { };&\node [block] (p1) {$\mathrm{update\: \transmitter}$}; & \node [text centered] (y1) {$\bot$};\\
				\node [text centered] (x2) {$\plainvec(1)$ };&\node [block] (p2) {$\mathrm{update\: \transmitter}$}; & \node [text centered] (y2) {$\ciphervec(1)=\ciphervec(0)$};\\
				\node [text centered] (x4) {$\plainvec(2)$ };&\node [block] (p5) {$\mathrm{update\: \transmitter}$}; & \node [text centered] (y4) {$\ciphervec(2)=\encryptionfun_{_{\key}}(., \plainvec(1))$};\\
				&\node [text centered] (p4) {$\cdots$}; & \\
				\node [text centered] (x3) { $\plainvec(t)=stop$};&\node [block] (p3) {$\mathrm{update\: \transmitter}$}; & \node [text centered] (y3) {$\ciphervec(t)=\encryptionfun_{_{\key}}(., \plainvec(t-1))$};&\\
			};
			%
			% legend for subprocedures
			\begin{scope} [every path/.style=line]
			\path (x1) -- node [] {} (p1);
			\path (p1) -- node [] {} (y1);
			\path (x2) -- node [] {} (p2);
			\path (x4) -- node [] {} (p5);
			\path (p5) -- node [] {} (y4);
			\path (p2) -- node [] {} (y2);
			\path (x3) -- node [] {} (p3);
			\path (p3) -- node [] {} (y3);
			\path (p1) -- node [] { \scriptsize{$(\state(1),\memvec(1))$}} (p2) ;
			\path (p2) -- node [] { \scriptsize{$(\state(2),\memvec(2))$}} (p5) ;
			\path (p5) -- node [] { \scriptsize{$(\state(3),\memvec(3))$}} (p4) ;
			\path (p4) -- node [] { \scriptsize{$(\state(t),\memvec(t))$}} (p3) ;
			\end{scope}
			\end{tikzpicture}
			\caption{ The general scheme of $\matr{S}_{\sigma}^4(\mathfrak{P})$ cryptosystem steps.
			}\label{fig:DSPLC2}
		\end{figure}

		\item{{\bf Decryption procedure ($\Dec_{\theta}$)}\\
			
			The input to the receiver is the ciphertext stream and the initial-vector $(\{\cipher(t)\}, \matr{ IV}).$
			In $\matr{S}_{\sigma}^4(\mathfrak{P})$ we have $t_{_{c}}=1$ as the cipher delay and the receiver operates as an unknown input observer as follows.
						
			\begin{itemize}
				\item The receiver selects a vector $\hstate(1)=(\hstate(1)_{1} \cdots,\hstate(1)_{n}) ^{T}$
				at random in $\Field_{_{q}}^{n}$.
				\item The general setup of the receiver procedure is as follows,
				{\scriptsize
						\[
						\receptor_{\kappa}: \left\lbrace
						\begin{array}
						{ll}
					\rm	Input:&(\{\ciphervec(t) \ t\geq1\}, \matr{IV})\\
					\rm Initial	:&\pi_{_{\key}},\sigma(t), \hstate(1), \matr{ IV}, \lbrace \Bmat_{j}\rbrace, \lbrace \Emat_{j}\rbrace, \lbrace \Fmat_{j}\rbrace, \lbrace \Lmat_{j}\rbrace, \Wmat, \memM\\
						\kernel\receptor (\ciphervec(t)), t\geq 1 &\\
						\ {\rm Output:}& \{\widehat{{\plainvec}}(t)\ t\geq 2\}.
						\end{array}
						\right.
					\]
					\[
					\kernel\receptor (\ciphervec(t)):\left\lbrace
					\begin{array}
					{ll}
					\matr{IV}=(\memvec(0),\mathbf{c}(0)),	\ \ \hstate(1) \overset{\$}{\leftarrow} \Field_{_{q}}^{n},\\
					if\ t=0\ then \\
					\ \ {\rm Output:} \ Ack\\
					if\ t=1\ then \\
					\ \ \left\lbrace
					\begin{array}	{l}
					\hstate(2)=\hstate(1),\\
					\widehat{\memvec}(2)=\memvec(0),\\
					\end{array}
					\right.\\
					\ \ {\rm Output:} \ Ack\\
					Else \\
					\ \ \left\lbrace
					\begin{array}	{l}\mathrm{update\: \receptor},\\
					\hplainvec(t)=\decryptionfun_{_{\key}}(\hkeystreamvec(t), \ciphervec(t), \lbrace \Fmat^{-1}_{j}\rbrace ),\\
					
					\end{array}
					\right.\\
					\ {\rm Output:} \ \hplainvec(t).\\
					\end{array}
					\right.
					\]

 				}
				\ \\ \ \\
				where the details of the functions in update procedure are as follows,
				{\scriptsize
					\begin{equation}\label{lab:DSPLC3}
						\begin{array}
							{llll}
							\mathrm{update\: \receptor}:\left\lbrace
							\begin{array}{llll}
											\rm 	Initial &: j=\sigma(t-1),\\
								\hkeystreamvec(t)&=\Lmat_{j}\hstate(t)+ \Bmat_{j}\BigPi_{_{\key}}(\hstate(t)),\\
								\hstate(t+1)&=\Wmat\widehat{\memvec}(t)+\Dmat_{j}\hstate(t)+ \Amat_{j}\BigPi_{_{\key}}(\hstate(t))+\Rmat_{j}(\ciphervec(t)-\hkeystreamvec(t)),\\
								\widehat{\memvec}(t+1)&=\memM\widehat{\memvec}(t)+ \ciphervec(t-1),\\
							\end{array}
							\right.
						\end{array}
					\end{equation}
					
					\begin{equation}\label{lab:DSPLC4}
						Output\ \receptor:\hplainvec(t)=\BigPi_{_{\key}}^{(-1)}(\Fmat_{j}^{-1}(\ciphervec(t)-\hkeystreamvec(t))).\\
					\end{equation}
				}
			\end{itemize}
			
		}
		
	} \end{itemize}

\subsection{Verification of system properties}

The following lemma describes the main properties of $\matr{S}_{\sigma}^4(\mathfrak{P})$.  

\begin{lem}
The $\matr{S}_{\sigma}^4(\mathfrak{P})$ cryptosystem is {\it finite-time self-synchronized} with $t_{_{s}}=n_{_{0}}$, delay $1$ and dummy ciphertext symbols $t_{_{c}}\leq t_{_{s}}+m_{_{0}}$. Also, $\transmitter_{\kappa}$ is a flat dynamical systems.
\end{lem}
\begin{proof}{
Clearly the lemma follows from the following four claims.

 		\begin{itemize}
			\item[{\rm i)}] The algorithm $\receptor_{\key}$ operates as an unknown input observer for $\transmitter_{\key}$ in a self-synchronized setup, i.e.
$$ \forall t \geq t_{_{s}}, \ \ \hstate(t+1)=\state(t).$$
			\item[{\rm ii)}]For any $ t \geq t_{_{s}} $ we have $\hplainvec(t+1)=\plainvec(t)$ and
system's delay is $d=1$.
			\item[{\rm iii)}]For any $ t \geq t_{_{s}}$, we have
			{\scriptsize
				\begin{equation}\label{lab:explicit}
					\begin{array}{ll}
						\state(t+1)=(\prod_{i=1}^{t}\Q_{\sigma(i)})\state(1)\\
						+\sum_{h=1}^{h=t}\left[(\prod_{i=h+1}^{t}\Q_{\sigma(i)})\Wmat\memvec(h)+(\prod_{i=h+1}^{t}\Q_{\sigma(i)})\Emat_{\sigma(h)}\Fmat_{\sigma(h)}^{-1} (\ciphervec(h)) \right].
					\end{array}
				\end{equation}
			}
	\item[{\rm iv)}] For all $ t\geq t_{_{s}}$, the vector
$ \state(t+1) $ and 	$ \hstate(t+2) $ depends on a finite number of previous ciphertexts, i.e.
{\scriptsize
	\begin{equation}
	\hstate(t+2)=\state(t+1)=\left\lbrace
	\begin{array}{ll}
	F_{_{\key}}(\ciphervec(t-t_{_{s}}-m_{_{0}}), \cdots, \ciphervec(t)) & t > m_{_{0}}\\
	F_{_{\key}}(m(0)_{t-t_{_{s}}} \cdots, m(0)_{m_{_{0}}},\ciphervec(1), \cdots, \ciphervec(t)) & t \leq m_{_{0}}.\\
	
	\end{array}\right.
	\end{equation}
	
}
			
		\end{itemize}
		
Now, we prove each on of these claims as follows.

 For Part~(i) note that,

 			{\scriptsize
				\begin{equation}\label{lab:DSPLC7}
					\begin{array}{ll}
						\ke=&\hstate((t+1)+1)-\state(t+1)\\&=
						\Wmat\widehat{\memvec}(t+1)+\Dmat_{j}\hstate(t+1)+ \Amat_{j}\BigPi_{_{\key}}(\hstate(t+1))+\Rmat_{j}(\ciphervec(t+1)-\hkeystreamvec(t+1))\\&
						\:-(\Wmat\memvec(t)+\Dmat_{j}\state(t)+ \Amat_{j}\BigPi_{_{\key}}(\state(t))+\Emat_{j}\BigPi_{_{\key}}(\plainvec(t)))\\&
						=\Dmat_{j}\ek+ \Amat_{j}\BigPi_{_{\key}}(\hstate(t+1))\\&
						\:+\Rmat_{j}(\Lmat_{j}\state(t)+ \Bmat_{j}\BigPi_{_{\key}}(\state(t))+\Fmat_{j}\BigPi_{_{\key}}(\plainvec(t))-\Lmat_{j}\hstate(t+1)- \Bmat_{j}\BigPi_{_{\key}}(\hstate(t+1)))\\&
						\:- \Amat_{j}\BigPi_{_{\key}}(\state(t))-\Emat_{j}\BigPi_{_{\key}}(\plainvec(t))\\&
						=(\Dmat_{j}-\Rmat_{j}\Lmat_{j})\ek=\Q_{j}\ek.
					\end{array}
				\end{equation}
			}
			Hence, since $\{ {\Q}_{j} \ |\ j \in [ \ell ]\}$ is a nilpotent semigroup of matrices of index $n_{_{0}}$,
			 $$ \forall t \geq t_{_{s}}=n_{_{0}}, \ \ \hstate(t+1)=\state(t).$$
			
Part~(ii) is a direct consequence of Part~(i), i.e.,
			{\scriptsize
				\begin{equation}
					\begin{array}{ll}
						\hplainvec(t+1)&=\BigPi_{_{\key}}^{(-1)}(\Fmat_{j}^{-1}(\ciphervec(t+1)-\hkeystreamvec(t+1)))\\
						&=\Fmat_{j}^{-1} \BigPi_{_{\key}}^{-1}(\ciphervec(t+1)-\Bmat_{j}\BigPi_{_{\key}}(\hstate(t+1))-\Lmat_{j}\hstate(t+1))\\
						&=\Fmat_{j}^{-1} \BigPi_{_{\key}}^{-1}(\Bmat_{j}\BigPi_{_{\key}}(\state(t))+\Lmat_{j}\state(t)+\Fmat_{j}\BigPi_{_{\key}}(\plainvec(t))-\Bmat_{j}\BigPi_{_{\key}}(\hstate(t+1))-\Lmat_{j}\hstate(t+1))\\
						&=\BigPi_{_{\key}}^{(-1)}(\Fmat_{j}^{-1}(\Fmat_{j}\BigPi_{_{\key}}(\plainvec(t))))=\plainvec(t).
					\end{array}
				\end{equation}	
			}
On the other hand, since $\ciphervec(t+1)=\encryptionfun_{_{\key}}(\keystreamvec(t), \plainvec(t))$ we deduce that the  delay is	$d=1$.
			
			For Part~(iii) we use induction on time $ t\geq 2$.
			\begin{itemize}
				\item {\it For  $t=2$}:
			
			Considering  $\ciphervec(t+1)=\keystreamvec(t)+\Fmat_{j}\BigPi_{_{\key}}(\plainvec(t))$, we have,
				{\scriptsize
					\begin{equation}\label{lab:2}
						\begin{array}{ll}
							\state(2)=&\Wmat\memvec(1)+\Dmat_{\sigma(1)}\state(1)+ \Amat_{\sigma(1)}\BigPi_{_{\key}}(\state(1))+\Emat_{\sigma(1)}\BigPi_{_{\key}}(\plainvec(1))\\&\Wmat\memvec(1)+\Dmat_{\sigma(1)}\state(1)+ \Amat_{\sigma(1)}\BigPi_{_{\key}}(\state(1))+\Emat_{\sigma(1)}\Fmat_{\sigma(1)}^{-1}\\&
							(\ciphervec(2)-\Bmat_{\sigma(1)}\BigPi_{_{\key}}(\state(1))-\Lmat_{\sigma(1)}\state(1))=\\&	(\Dmat_{\sigma(1)}-\Emat_{\sigma(1)}\Fmat_{\sigma(1)}^{-1}\Lmat_{\sigma(1)})\state(1)+\Wmat\memvec(1)+\Emat_{\sigma(1)}\Fmat_{\sigma(1)}^{-1} \ciphervec(2).
						\end{array}
					\end{equation}
				}
				
				\item {\it The induction step for time $t$}:\\
				Assuming (\ref{lab:explicit}); that is,
				{\scriptsize
					
					\begin{equation}\label{lab:explicit2}
						\begin{array}{l}
							\state(t)=(\prod_{i=1}^{t-1}\Q_{\sigma(i)})\state(1)+\\\sum_{h=1}^{h=t-1}\left[(\prod_{i=h+1}^{t-1}\Q_{\sigma(i)})\Wmat\memvec(h)+(\prod_{i=h+1}^{t-1}\Q_{\sigma(i)})\Emat_{\sigma(h)}\Fmat_{\sigma(h)}^{-1} \ciphervec(h+1)\right].
						\end{array}
					\end{equation}
					
				}
				then using equations (\ref{lab:DSPLC}) and (\ref{lab:DSPLC2}) for 
				$ \state(t+1) $ we may conclude that,
				{\scriptsize
					
					\begin{equation}\label{lab:explicit3}{
	\begin{array}{lllll}
					\state(t+1)&=\Wmat\memvec(t)+\Amat_{\sigma(t)}\BigPi_{_{\key}}(\state(t))+\Emat_{\sigma(t)}\BigPi_{_{\key}}(\plainvec(t))+\Dmat_{\sigma(t)}\state(t)\\&
					=\Wmat\memvec(t)+\Amat_{\sigma(t)}\BigPi_{_{\key}}(\state(t))+\Emat_{\sigma(t)}\BigPi_{_{\key}}(\plainvec(t))+\Dmat_{\sigma(t)}((\prod_{i=1}^{t-1}\Q_{\sigma(i)})\state(1)\\&
					+\sum_{h=1}^{h=t-1}\left[(\prod_{i=h+1}^{t-1}\Q_{\sigma(i)})\Wmat\memvec(h)+(\prod_{i=h+1}^{t-1}\Q_{\sigma(i)})\Emat_{\sigma(h)}\Fmat_{\sigma(h)}^{-1}\ciphervec(h+1)\right])\\&
					=(\prod_{i=1}^{t}\Q_{\sigma(i)})\state(1)+\\&
					\sum_{h=1}^{h=t}\left[(\prod_{i=h+1}^{t}\Q_{\sigma(i)})\Wmat\memvec(h)+(\prod_{i=h+1}^{t}\Q_{\sigma(i)})\Emat_{\sigma(h)}\Fmat_{\sigma(h)}^{-1} \ciphervec(h+1)\right].
					\end{array}
	}\end{equation}
					
				}
				Thus, (\ref{lab:explicit}) holds for time $(t+1) $, and the proof of the induction step is complete.
				
			\end{itemize}
			For Part~(iv) note that $\{ {\Q}_{j} \ |\ j \in [ \ell ]\}$
		is a nilpotent semigroup of matrices of index $t_{_{s}}=n_{_{0}}$, and consequently, for $t \geq t_{_{s}}$ we have 			$(\prod_{i=1}^{t}\Q_{\sigma(i)})\state(1)=0$. Hence, using (\ref{lab:explicit3}), for any $ t\geq t_{_{s}}$ we have
		{\scriptsize
			\begin{equation}\label{eq:s}
			\begin{array}{ll}
			\state(t+1)=\sum_{h=t-t_{_{s}}}^{h=t}\left[(\prod_{i=h+1}^{t}\Q_{\sigma(i)})\Wmat\memvec(h)+
			(\prod_{i=h+1}^{t}\Q_{\sigma(i)})\Emat_{\sigma(h)}\Fmat_{\sigma(h)}^{-1}\ciphervec(h) \right].
			\end{array}
			\end{equation}
		}
		Also, since for $h>n$, and $\memvec(h)$ is a linear function of $\ciphervec(h-1-m_{_{0}}), \cdots, \ciphervec(h-1)$ and for  $h\leq  m_{_{0}}$, we know that $\memvec(h)$  is a linear function of $\ciphervec(1), \cdots, \ciphervec(h-1)$ and $m(0)_{h} \cdots, m(0)_{m_{_{0}}} $ we may conclude that for all
		$ t \geq t_{_{s}} $ and ,
		
		\begin{equation}
		\state(t+1)=\left\lbrace
		\begin{array}{ll}
		F_{_{\key}}(\ciphervec(t-t_{_{s}}-m_{_{0}}), \cdots, \ciphervec(t)) & t > m_{_{0}}\\
		F_{_{\key}}(m(0)_{t-t_{_{s}}} \cdots, m(0)_{m_{_{0}}},\ciphervec(1), \cdots, \ciphervec(t)) & t \leq m_{_{0}}.\\
		
		\end{array}\right.
		\end{equation}
		
		The amount of memory needed at most to save the necessary ciphertexts from the past is $$t-(t-t_{_{s}}-m_{_{0}})=t_{_{s}}+m_{_{0}},$$ and consequently, $\transmitter_{\kappa}$ is flat.
		
		Now, we prove that for $ t\geq t_{_{s}}$	
		\begin{equation}
		\hstate(t+2)=\left\lbrace
		\begin{array}{ll}
		F_{_{\key}}(\ciphervec(t-t_{_{s}}-m_{_{0}}), \cdots, \ciphervec(t)) & t > m_{_{0}},\\
		F_{_{\key}}(m(0)_{t-t_{_{s}}} \cdots, m(0)_{m_{_{0}}},\ciphervec(1), \cdots, \ciphervec(t)) & t \leq m_{_{0}}.\\
		
		\end{array}\right.
		\end{equation}
		We use induction on time $ t\geq 3$.
		\begin{itemize}
			\item {\it For $t=3$}:
			
			We have,
			{\scriptsize
				\begin{equation}\label{lab:22}
				\begin{array}{ll}
				\hstate(3)=&\Wmat\widehat{\memvec}(2)+\Dmat_{\sigma(1)}\hstate(2)+ \Amat_{\sigma(1)}\BigPi_{_{\key}}(\hstate(2))+\Rmat_{\sigma(1)}(\ciphervec(2)-\hkeystreamvec(2))\\&=	(\Dmat_{\sigma(1)}-\Emat_{\sigma(1)}\Fmat_{\sigma(1)}^{-1}\Lmat_{\sigma(1)})\hstate(2)+\Wmat\widehat{\memvec}(2)+\Rmat_{\sigma(1)} \ciphervec(2).
				\end{array}
				\end{equation}
			}
			
			\item {\it The induction step for time $t$}:\\
			Assuming (\ref{lab:explicit}); that is,
			{\scriptsize
				
				\begin{equation}\label{lab:explicit22}
				\begin{array}{l}
				\hstate(t+1)=(\prod_{i=1}^{t-1}\Q_{\sigma(i)})\hstate(2)+\\\sum_{h=1}^{h=t-1}\left[(\prod_{i=h+1}^{t-1}\Q_{\sigma(i)})\Wmat\widehat{\memvec}(h+1)+(\prod_{i=h+1}^{t-1}\Q_{\sigma(i)})\Rmat_{\sigma(h)} \ciphervec(h+1)\right].
				\end{array}
				\end{equation}
				
			}
			then using equations (\ref{lab:DSPLC}) and (\ref{lab:DSPLC2}) for
			$ \state(t+1) $ we may conclude that,
			{\scriptsize
				
				\begin{equation}\label{lab:explicit33}{
					\begin{array}{lllll}
					\hstate(t+2)&=\Wmat\widehat{\memvec}(t+1)+\Dmat_{\sigma(t-1)}\hstate(t+1)+\Amat_{\sigma(t-1)}\BigPi_{_{\key}}(\hstate(t+1))+\Rmat_{\sigma(t-1)}(\ciphervec(t+1)-\hkeystreamvec(t+1))
					\\&=(\prod_{i=1}^{t}\Q_{\sigma(i)})\hstate(2)+
					\sum_{h=1}^{h=t}\left[(\prod_{i=h+1}^{t}\Q_{\sigma(i)})\Wmat\widehat{\memvec}(h+1)+(\prod_{i=h+1}^{t}\Q_{\sigma(i)})\Rmat_{\sigma(h)} \ciphervec(h+1)\right]
					\end{array}
				}\end{equation}
				
			}		
		\end{itemize}
		
		Since $\{ {\Q}_{j} \ |\ j \in [ \ell ]\}$
		is a nilpotent semigroup of matrices of index $t_{_{s}}=n_{_{0}}$, and consequently, for $t \geq t_{_{s}}$ we have 	
		$(\prod_{i=1}^{t}\Q_{\sigma(i)})\hstate(2)=0$. Hence, using (\ref{lab:explicit3}), for any $ t\geq t_{_{s}}$ we have
		{\scriptsize
			\begin{equation}{\label{eq:shat}
				\hstate(t+2)=\sum_{h=t-t_{_{s}}}^{h=t}\left[(\prod_{i=h+1}^{t}\Q_{\sigma(i)})\Wmat\widehat{\memvec}(h+1)+
				(\prod_{i=h+1}^{t}\Q_{\sigma(i)})\Emat_{\sigma(h)}\Fmat_{\sigma(h)}^{-1}\ciphervec(h) \right].
			}\end{equation}
		}
		
		Also, since $\widehat{\memvec}(h+1)={\memvec}(h)$ then
		
		\begin{equation}
		\hstate(t+2)=\left\lbrace
		\begin{array}{ll}
		F_{_{\key}}(\ciphervec(t-t_{_{s}}-m_{_{0}}), \cdots, \ciphervec(t)) & t > m_{_{0}},\\
		F_{_{\key}}(m(0)_{t-t_{_{s}}} \cdots, m(0)_{m_{_{0}}},\ciphervec(1), \cdots, \ciphervec(t)) & t \leq m_{_{0}},\\
		
		\end{array}\right.
		\end{equation}
		as we wanted to show.
}\end{proof}

		\begin{exm}{\label{exm:sync}{\bf An illustrative example}\\

		We consider a 3-dimensional $\matr{S}^4_{\sigma}$ cryptosystem with following matrices (see Equations~\ref{lab:DSPLC} and \ref{lab:DSPLC1})
		with $\state(t)\in\f^{3}_{7}$, $\plainvec(t)\in\f^{3}_{7}$ and $\ciphervec(t)\in\f^{3}_{7}$ and applying a simple swiching function
		$ \sigma(t)=t $\ \texttt{mod }$ 2$;
		{\scriptsize

			\[
	\mathbf{Q_1}=\left(\begin{array}{*3{c}}
			6&1&0\\
			6&1&0\\
			0&0&0
			\end{array}\right)
			\mathbf{D_1}=\left(\begin{array}{*3{c}}
			6&0&2\\
			6&3&3\\
			0&3&2
			\end{array}\right)
			\mathbf{A_1}=\left(\begin{array}{*3{c}}
			0&2&5\\
			6&1&1\\
			1&0&1
			\end{array}\right)
			\mathbf{E_1}=\left(\begin{array}{*3{c}}
			0&6&0\\
			0&0&2\\
			5&1&0
			\end{array}\right)
			\mathbf{L_1}=\left(\begin{array}{*3{c}}
			1&1&5\\
			0&3&1\\
			0&1&0
			\end{array}\right)\]
			\[\mathbf{B_1}=\left(\begin{array}{*3{c}}
			0&5&2\\
			3&0&1\\
			0&6&0
			\end{array}\right)
			\mathbf{F_1}=\left(\begin{array}{*3{c}}
			0&1&0\\
			0&2&1\\
			3&0&4
			\end{array}\right)
			\mathbf{Q_2}=\left(\begin{array}{*3{c}}
			2&5&0\\
			2&5&0\\
			0&0&0
			\end{array}\right)
			\mathbf{D_2}=\left(\begin{array}{*3{c}}
			0&3&5\\
			2&0&4\\
			5&4&6
			\end{array}\right)
			\mathbf{A_2}=\left(\begin{array}{*3{c}}
			2&6&0\\
			5&0&4\\
			3&1&1
			\end{array}\right)
			\]
			\[\mathbf{E_2}=\left(\begin{array}{*3{c}}
			3&1&0\\
			1&0&1\\
			0&1&2
			\end{array}\right)
			\mathbf{L_2}=\left(\begin{array}{*3{c}}
			1&6&1\\
			0&0&1\\
			0&1&0
			\end{array}\right)
			\mathbf{B_2}=\left(\begin{array}{*3{c}}
			1&2&0\\
			0&3&1\\
			6&1&0
			\end{array}\right)
			\mathbf{F_2}=\left(\begin{array}{*3{c}}
			0&3&0\\
			5&0&2\\
			1&0&0
			\end{array}\right)
	,\mathbf{W}=\left(\begin{array}{*3{c}}
			0&1&0\\
			0&2&1\\
			3&0&4
			\end{array}\right),\]\[
			\mathbf{M}=\left(\begin{array}{*3{c}}
			0&1&0\\
			0&0&1\\
			0&0&0
			\end{array}\right)
			\]
		}
			Figure~\ref{fig:example} depicts the error vector between the $\plainvec(t)$ (regular line) and $\hplainvec(t)$ (dotted line) with $\state(0)=(2,4,1)^{T}$, $\hstate(0)=(0,2,4)^{T}$, $\ciphervec=(1,4,4)^T$ and $\memvec=(0,0,0)^T$. Note that when $t\geq t_{_{s}}=2$ we have $\hstate(t+1)=\state(t)$ and
		after two clocks both sequences are synchronized. 
	}\end{exm}

		\begin{figure}[t]
	
		\centering
	  	\includegraphics[width=80mm]{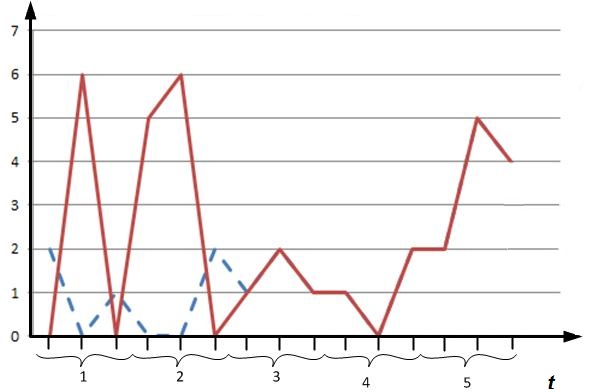}
		\caption{System synchronization (see Example~\ref{exm:sync})}
	 \label{fig:example}
		\end{figure}

\section{Security analysis}\label{sec:cryptoanalysis}

In this section we analyze the security of $\matr{S}_{\sigma}^4(\mathfrak{P})$   
in the LORBACPA$^+$ model. In what follows we always assume that the secret key, $\key$,
is a $k$-bit string.
Also, $\mathfrak{U}$ refers to the ensemble of truly (i.e. uniformly) random permutations of $\Field_{q}$, while $\mathfrak{P}$ is a pseudorandom emsemble of permutations used
in the cryptosystem $\matr{S}_{\sigma}^4(\mathfrak{P})$.

Recall that $\matr{S}_{\sigma}^4(\mathfrak{P})$ is said to be LORBACPA$^+$ secure, if for all probabilistic polynomial-time adversaries $\mathcal{A}$ (as in Definition~\ref{def:exp}),
$$\Insec^{LORBACPA^+}_{{\matr{S}_{\sigma}^4(\mathfrak{P})}}(k)=
\displaystyle{\max_{{\mathcal A}}}\
\Adv^{LORBACPA^+}_{ \mathcal{A},{\matr{S}_{\sigma}^4(\mathfrak{P})}}
\leq \negl(k).
$$	
Before we proceed, let us fix the setup. Assuming the existence of an adversary $\mathcal{A}$ in the security model LORBACPA$^+$, then $1 \leq i \leq s$ is the indicator of query sessions where it is assumed that the whole number of sessions is equal to $s$. In this setting $s_{_{E}}$ is the number of query sessions initiated by the oracle $E$ and $s_{_{SE}}$ is the corresponding number for the oracle $SE$.
Similarly, $\nu_{_{E}}$ (resp. $\nu_{_{SE}}$) is the whole number of queries 
asked from the oracle $E$ within $s_{_{E}}$ sessions (resp. $SE$ within $s_{_{SE}}$ sessions).

Note that each one of $\nu_{_E}$ queries from $E$ made by the adversary $\mathcal{A}$ consists of a pair of equal length messages $\plainvec_{i,E}^0(t)$ and $\plainvec_{i,E}^1(t)$ as vectors in $\Field_{_{q}}^{n}$ along with the vector $\matr{IV}_{i}$ related to initialization of session $i$.
 For each $1 \leq i \leq \nu_{_E}$  the vector $(\ciphervec_{i,E}^{b}(\tau), \matr{IV}_{i})$ stands for the answer of the oracle $E$ to the query
$(\plainvec^0_{i,E}(\tau),\plainvec^1_{i,E}(\tau),\matr{IV}_{i})=q(E,i,\tau,\nu)$ for which $\state_{i}(\tau)$ is the internal state of the
oracle $E$ simulating the encryption scheme in session $i$ with the memory vector $\memvec_{i}(\tau)$. 
Similarly,
for each $1 \leq i \leq \nu_{_{SE}}$ the vector $(\ciphervec_{i,SE}^{b}(\tau), \matr{IV}_{i})$ stands for the answer to the query $(\plainvec^0_{i,SE}(\tau),\plainvec^1_{i,SE}(\tau),\matr{IV}_{i})$
from the oracle $SE$ with the internal state $\hstate_{i}(\tau)$ and the memory vector $\widehat{\memvec}_{i}(\tau)$.
Recall that in this setting $\tau$ is the local counter dedicated to each session while 
$\nu$ is the global query counter of the algorithm.

Also, note that for $b \in \{0,1\}$, within the oracle $E(\matr{IV},b)$
we have 
\begin{equation}\label{eq:psi1}
\begin{array}{l}
\quad \ciphervec_{i,E}^{b}(\tau+1)=\Lmat\state_{i}(\tau)+ \Bmat\BigPi_{_{\key}}(\state_{i}(\tau))+\Fmat\BigPi_{_{\key}}({\plainvec^b}_{i,E}(\tau)),
\end{array}
\end{equation}
and within $SE(\matr{IV},b)$ we have 
\begin{equation}\label{eq:psi2}
\begin{array}{l}
\quad \ciphervec_{i,SE}^{b}(\tau+1)=\Lmat\hstate_{i}(\tau+1)+ \Bmat\BigPi_{_{\key}}(\hstate_{i}(\tau+1))+\Fmat\BigPi_{_{\key}}({\plainvec^b}_{i,SE}(\tau)).
\end{array}
\end{equation}
As an standard stage of the security proof, we first reduce the security of $\matr{S}_{\sigma}^4(\mathfrak{P})$ to that of $\matr{S}_{\sigma}^4(\mathfrak{U})$ in the following 
proposition.

\begin{pro} \label{pro:randomcase}
	Considering two cryptosystems $\matr{S}_{\sigma}^4(\mathfrak{P})$ and
	$\matr{S}_{\sigma}^4(\mathfrak{U})$ with the same set of parameters, then
	\[|\Adv^{LORBACPA^+}_{ \mathcal{A},{\matr{S}_{\sigma}^4(\mathfrak{P})}}(k) - \Adv^{LORBACPA^+}_{ \mathcal{A},{\matr{S}_{\sigma}^4(\mathfrak{U})}}(k)| < \negl(k).\]
\end{pro}
\begin{proof}{
		
		Consider an adversary $\mathcal{A}$ against $\matr{S}_{\sigma}^4(\mathfrak{P})$ within the setting of LORBACPA$^+$ security model. Using this adversary, we construct a distinguisher $\mathcal{D}$, for $\mathfrak{P}$.
		
		More precisely, the distinguisher $\mathcal{D}$ interacts
		with a permutation oracle $O_{\pi}$, that in the beginning of the game, flips a random bit $b$ and if $b = 0$, chooses a random permutation $\pi \leftarrow \mathfrak{U}$, while otherwise, if $b = 1$, the oracle chooses a permutation $\pi \leftarrow \mathfrak{P}$.
		
		To this end, since $\mathcal{D}$ uses the adversary $\mathcal{A}$ as a subroutine, it has to simulate the environment of the adversary $\mathcal{A}$.
		For this, let us describe how $\mathcal{D}$ answers the queries made by $\mathcal{A}$. The distinguisher $\mathcal{D}$ runs $\mathcal{A}$ and has to concurrently answer	the block queries used to be answered by the oracles $SE(\matr{IV},b)$ or $E(\matr{IV},b)$. In detail, $\mathcal{D}$ works as follows (see Algorithm~\ref{alg:distinguisher}):\\ 
		
		\textbf{Distinguisher $\mathcal{D}$}:\\
		
		The input of $\mathcal{D}$ is $1^k$ and has access to a permutation oracle $O_{\pi}$.
		\begin{itemize}
			\item[1-] \textbf{ Initialization}: $\mathcal{D}$ picks at random a bit $b \in \{0, 1\}$ chosen uniformly at random and sets $O_{\pi}$.
			
			\item[2-] \textbf{Running $\mathcal{A}$ and answering queries}: $\mathcal{D}$ runs $\mathcal{A}$, answering queries via the subroutine $LR$ described below. 
			\begin{itemize}
				\item[-] \textbf{$\mathcal{A}$ feeds}: $\mathcal{A}$ can submit queries $(\plainvec^{0}_{i,E}(\tau), \plainvec^{1}_{i,E}(\tau),\matr{IV}_{i})$ or $(\plainvec^{0}_{i',SE}(\tau), \plainvec^{1}_{i',SE}(\tau),\matr{IV}_{i'})$ where $i$ and $i'$ are session indicators.
				\item[-] \textbf{ $\mathcal{D}$ answers}:  If $\mathcal{D}$ receives  a query $(\plainvec^{0}_{i,E}(\tau), \plainvec^{1}_{i,E}(\tau),\matr{IV}_{i})$ it returns $(\ciphervec_{i,E}^{b}(\tau), \matr{IV}_{i})$ to $\mathcal{A}$, and if it receives a query $(\plainvec^{0}_{i',SE}(\tau), \plainvec^{1}_{i',SE}(\tau),\matr{IV}_{i'})$ it returns $(\ciphervec_{i',SE}^{b}(\tau), \matr{IV}_{i'})$ to $\mathcal{A}$. The key points are as follows:
				\begin{itemize}
					\item[.] If $b=1$, $\mathcal{D}$'s oracle uses $\pi \leftarrow \mathfrak{P}$ (a pseudorandom permutation). In this case if $\mathcal{D}$ receives query $(\plainvec^{0}_{i,E}(\tau), \plainvec^{1}_{i,E}(\tau),\matr{IV}_{i})$ it always simulates the oracle $E(\matr{IV},1)$ for $\matr{S}_{\sigma}^4(\mathfrak{P})$ and
					returns $(\ciphervec_{i,E}^{b}(\tau), \matr{IV}_{i})$ to $\mathcal{A}$. Also, if a query $(\plainvec^{0}_{i',SE}(\tau), \plainvec^{1}_{i',SE}(\tau),\matr{IV}_{i})$ is received, then  $SE(\matr{IV},1)$ is simulated for $\matr{S}_{\sigma}^4(\mathfrak{P})$ and  $(\ciphervec_{i',SE}^{b}(\tau), \matr{IV}_{i'})$ is returned  to $\mathcal{A}$.
					\item[.] If $b=0$, $\mathcal{D}$ 's oracle uses $\pi\leftarrow \mathfrak{U}$ (a truly random
					permutation).
					In this case if $\mathcal{D}$ receives query $(\plainvec^{0}_{i,E}(\tau), \plainvec^{1}_{i,E}(\tau),\matr{IV}_{i})$ it always simulates the oracle $E(\matr{IV},0)$ for $\matr{S}_{\sigma}^4(\mathfrak{U})$ and
					returns $(\ciphervec_{i,E}^{b}(\tau), \matr{IV}_{i})$ to $\mathcal{A}$. Also, if a query $(\plainvec^{0}_{i',SE}(\tau), \plainvec^{1}_{i',SE}(\tau),\matr{IV}_{i})$ is received, then  $SE(\matr{IV},0)$ is simulated for $\matr{S}_{\sigma}^4(\mathfrak{U})$ and  $(\ciphervec_{i',SE}^{b}(\tau), \matr{IV}_{i'})$ is returned  to $\mathcal{A}$.
				\end{itemize}
			\end{itemize}
			\item[2-]\textbf{Final stage}: Continue answering any oracle queries of $\mathcal{A}$ as described above, and at the end of the game, let $b'$ be the output of $\mathcal{A}$. Then, $\mathcal{D}$ outputs $1$ if $b=b'$ and outputs $0$ otherwise.
			
		\end{itemize}
			
		We have
		\[\begin{array}{ll}
		\Adv^{LORBACPA^+}_{ \mathcal{A},{\matr{S}_{\sigma}^4(\mathfrak{P})}}(k)& =2| Pr(output({\mathcal A})=1)-\dfrac{1}{2}|\\&=2| Pr(b' = b\mid b=1)-\dfrac{1}{2}|.
		\end{array}
		\]
		and
		\[\begin{array}{ll}
		\Adv^{LORBACPA^+}_{ \mathcal{A},{\matr{S}_{\sigma}^4(\mathfrak{U})}}(k)& =2| Pr(output({\mathcal A})=1)-\dfrac{1}{2}|\\&=2| Pr(b' = b\mid b=0)-\dfrac{1}{2}|.
		\end{array}
		\]
		Since we know that  ${\bf Adv}_{\mathfrak{P};D}(k)$ is a negligible function of the security parameter $k$, we have
		\[\begin{array}{ll}
		{\bf Adv}_{\mathfrak{P};D}(k) &=2| Pr(output({\mathcal D})=1)-\dfrac{1}{2}|\\
		&=| Pr(b' = b\mid b = 1)-Pr(b' =b\mid b = 0)|\\
		&=\dfrac{1}{2} |{\bf Adv}^{LORBACPA^+}_{ \mathcal{A},{\matr{S}_{\sigma}^4(\mathfrak{P})}}(k)- {\bf Adv}^{LORBACPA^+}_{ \mathcal{A},{\matr{S}_{\sigma}^4(\mathfrak{U})}}(k)| < \negl(k).
		\end{array}
		\]
}\end{proof}

	\begin{algorithm}[t]
	\caption{Distinguisher $\mathcal{D}$ (with access to $O_{\pi}$)}
	\label{alg:distinguisher}
	\begin{multicols}{2}
		\begin{algorithmic}
			\State
			\Procedure{Initialize}{}
			\State $b\overset{\$}{\leftarrow} \{0, 1\}$
			\EndProcedure
			\\
			\Procedure{run $\mathcal{A}$}{}
			\State $b'\overset{\$}{\leftarrow} \mathcal{A}^{LR} $
			\EndProcedure
			
			\\
			\Procedure{Finalization}{}
			\State Output $1$ if $b'=b$,
			\State Output $0$ otherwise.
			\EndProcedure
			\pagebreak
			
			\Procedure{$LR(\plainvec^{0}_{i,type}(t), \plainvec^{1}_{i,type}(t),\matr{IV}_{i})$}{}
			{\scriptsize
				\If {$type=E$}
				\State Set $E(\matr{IV},b)$ thanks to $O_{\pi}$
				\State $(\plainvec^{0}_{i,E}(t), \plainvec^{1}_{i,E}(t),\matr{IV}_{i}) \rightarrow E(\matr{IV}_{i},b)$
				\State $(\ciphervec_{i,E}^{b}(t), \matr{IV}_{i}) \leftarrow E(\matr{IV}_{i},b)$
				
				\Return $(\ciphervec_{i,E}^{b}(t), \matr{IV}_i)$
				\EndIf
				\If {$type=SE$}
				\State Set $SE(\matr{IV},b)$ thanks to $O_{\pi}$
				\State $(\plainvec^{0}_{i,SE}(t), \plainvec^{1}_{i,SE}(t),\matr{IV}_{i}) \rightarrow SE(\matr{IV}_{i},b)$
				\State $(\ciphervec_{i,SE}^{b}(t), \matr{IV}_{i}) \leftarrow SE(\matr{IV}_{i},b)$
				
				\Return $(\ciphervec_{i,SE}^{b}(t), \matr{IV}_i).$
				
				\EndIf
			}
			\EndProcedure
		\end{algorithmic}
	\end{multicols}
\end{algorithm}

%%%%%%%%%%%%%%%%%%%%%%%%%%%%%%%%%%%%%%%%%%%%%%%%%%%%%%%%%%%%%%%%%%%%%%%%%%%%%%%%%%%%%%%

Let $A$ be an adversary in the LORBACPA$^+$ security model as a randomized algorithm having interactions with LORBA encryption oracles $E$ and $SE$.
In this setting, let 
{\footnotesize
	$$ (\plainvec_{_{1}}^0,\plainvec_{_{1}}^1,i_{_{q_{_{1}}}})=q_{_{1}}(o_{_{1}},i_{_{q_{_{1}}}},\tau_{_{q_{_{1}}}},1),\  (\plainvec_{_{2}}^0,\plainvec_{_{2}}^1,i_{_{q_{_{2}}}})=q_{_{2}}(o_{_{2}},i_{_{q_{_{2}}}},\tau_{_{q_{_{2}}}},2),\ 
	\ldots,\ 
	(\plainvec_{_{r}}^0,\plainvec_{_{r}}^1,i_{_{q_{_{r}}}})=q_{_{r}}(o_{_{r}},i_{_{q_{_{r}}}},\tau_{_{q_{_{r}}}},r),$$
}
be the consecutive $r$ queries that the adversary asks from its oracles, and also let 
{\footnotesize
	$$\ciphervec^{b}(1)=\varepsilon_{_{\theta}}(\state(\tau_{_{q_{_{1}}}}+\delta_{o_{_{1}}}),\plainvec^b_{_1}(\tau_{_{q_{_{1}}}})), \ 
	\ciphervec^{b}(2)=\varepsilon_{_{\theta}}(\state(\tau_{_{q_{_{2}}}}+\delta_{o_{_{2}}}),\plainvec^b_{_2}(\tau_{_{q_{_{2}}}})), \
	\ldots, \ 
	\ciphervec^{b}(r)=\varepsilon_{_{\theta}}(\state(\tau_{_{q_{_{r}}}}+\delta_{o_{_{r}}}),\plainvec^b_{_r}(\tau_{_{q_{_{r}}}})),
	$$
}
be the corresponding consecutive answers, with $\delta_{_{SE}}=1$ and $\delta_{_E}=0$. Note that in this setting and for a fixed session $i$,
$l$ consecutive queries of session $i$ can be described as
$$ (\plainvec_{_{1}}^0,\plainvec_{_{1}}^1,i)=q_{_{j_1}}(o_{_{j_1}},i,1,\nu_{_{1}}),\ 
(\plainvec_{_{2}}^0,\plainvec_{_{2}}^1,i)=q_{_{j_2}}(o_{_{j_2}},i,2,\nu_{_{2}}),\ 
\ldots,
(\plainvec_{_{l}}^0,\plainvec_{_{l}}^1,i)=q_{_{j_{_{l}}}}(o_{_{j_{_{l}}}},i,l,\nu_{_{l}}),\ 
$$
where the answer to the $\tau$;th query $q_{_{j_\tau}}(o_{_{j_\tau}},i,\tau,\nu_{_{\tau}})=(\plainvec_{_{\tau}}^0,\plainvec_{_{\tau}}^1,i)$
is
\begin{equation}\label{eq:psi3}
\ciphervec^{b}(\nu_{_{\tau+1}})-\Fmat\BigPi_{_{\key}}(\plainvec^b_{_{\nu_{_{\tau}}}}(\tau)) \isdef
\Psi^b_{o}(\ciphervec^{b}(\nu_{_{\tau-1-\mu_{_{0}}}}),\ldots,\ciphervec^{b}(\nu_{_{\tau-1}})),
\end{equation}
where $\Psi^b_{o}$ is defined using Equations~\ref{eq:psi1} and \ref{eq:psi2}.
Moreover, the probability space is generated by the random bits used in the experiment for 
choosing the bit $b$, the key $\kappa$ (containing the information necessary to reconstruct the secret matrices and the secret random permutation) and the initial states of the oracles
$$\state(\delta_{o_{_{1}}}), \state(\delta_{o_{_{2}}}), \ldots, \state(\delta_{o_{_{r}}}),$$
 as well as the random bits used by the adversary (i.e. the randomized algorithm), 
 Also, note that $r$ is bounded by a polynomial of the security parameter $k$ since the adversary is a polynomial-time algorithm.

We define the following partial order on the sequence of queries (and consequently on answers),
$$ q_{_{u}}(o_{_{u}},i_{_{q_{_{u}}}},\tau_{_{q_{_{u}}}},\nu_{_{q_{_{u}}}}) \leq q_{_{v}}(o_{_{v}},i_{_{q_{_{v}}}},\tau_{_{q_{_{v}}}},\nu_{_{q_{_{v}}}}) 
\ \Leftrightarrow \ \left(i_{_{q_{_{u}}}}=i_{_{q_{_{v}}}} \ \& \ \tau_{_{q_{_{u}}}} \leq \tau_{_{q_{_{v}}}}\right).$$
Note that the inequality also implies that $\nu_{_{q_{_{u}}}} \leq \nu_{_{q_{_{v}}}}$ where, clearly, the partial order turns the set of queries (and consequently on answers)
into a well-founded (i.e. Noetherian) set.

Now, if $\nu > 0$, let $Col_{_{\nu}}$	be the event that for some $j_{_{1}} \leq \nu$ and $j_{_{2}} \leq \nu$ we have $\ciphervec^{b}(j_{_{1}})=\ciphervec^{b}(j_{_{2}})$ for some $b \in \{0,1\}$. Also, define $Col \isdef Col_{_{r}}$.

 Let $\xi \isdef (\ciphervec_{_{0}}, \cdots, \ciphervec_{_{\nu}}) $ be fixed sequence of blocks, 
 let $q_{_{\nu}}(o_{_{\nu}},i,\tau,\nu)$ be the $\nu$th query (i.e. $\nu \isdef \nu_{_{\tau}}$), define $\mu_{_{0}} \isdef m_{_{0}}+n_{_{0}}=m_{_{0}}+t_{_{s}}$ and also let
 $$q_{_{\nu_{_{\tau-1-\mu_{_{0}}}}}}(o_{_{\nu_{_{\tau-1-\mu_{_{0}}}}}},i,\tau-1-\mu_{_{0}},\nu_{_{\tau-1-\mu_{_{0}}}}),\ldots,q_{_{\nu_{_{\tau-1}}}}(o_{_{\nu_{_{\tau-1}}}},i,\tau-1,\nu_{_{\tau-1}}),$$
 be $\mu_{_{0}}+1$ of its consecutive predecessors (in the $i$th session)  with the following vector of answers,
 $$\varrho^b=(\ciphervec^{b}(\nu_{_{\tau-1-\mu_{_{0}}}}),\ldots,\ciphervec^{b}(\nu_{_{\tau-1}})).$$ 

Define $H_{_{\nu}}(\xi) \isdef A^0_{_{\nu}}(\xi) \cup A^1_{_{\nu}}(\xi)$ in which 
$$A^b_{_{\nu}}(\xi) \isdef \{ ({\bf x}_{_{\mu_{_{0}}}}, \cdots, {\bf x}_{_{0}})  \ | \ \exists\  j \leq \nu-1, \ \   \ciphervec_{_{j}}-\Fmat\BigPi_{_{\key}}(\plainvec^{b}_{_{j}})=\Psi^b_{o}({\bf x}_{_{0}}, \cdots, {\bf x}_{_{\mu_{_{0}}}})\}.$$

Our main objective in what follows is to prove that not only the probability of the event $Col$ is negligible but also the probability of success for the adversary conditioned to having no collision is also a negligible function of the security parameter. 
Formally, we have to prove the following statements.

\begin{pro} \label{thm:main}
	Let $\mathcal{A}$ be an adversary for $ \matr{S}_{\sigma}^4(\mathfrak{U})$ within the setting of LORBACPA$^+$ security model. Then for any $b \in \{0,1\}$ there exist a negligible function $\negl_{_{0}}$ such that 
	\begin{itemize}
		
		\item[{\rm a)}] $Pr_0(output({\mathcal A})=1|\overline{Col})=Pr_1(output({\mathcal A})=1|\overline{Col})$.		
		\item[{\rm b)}] $Pr_{0}(Col)=Pr_{1}(Col)$.
		\item[{\rm c)}] $Pr_b(Col) \leq \negl_{_{0}}(k).$
	\end{itemize}
\end{pro}

\begin{proof}{
       {\bf (a) $\Rightarrow$} For this part,  let $\xi \isdef (\ciphervec_{_{0}}, \cdots, \ciphervec_{_{\nu}}) $ be fixed sequence of answers, let $q_{_{\nu}}(o_{_{\nu}},i,\tau,\nu)$ be the $\nu$th query (i.e. $\nu \isdef \nu_{_{\tau}}$) and also let
		$$q_{_{\nu_{_{\tau-1-\mu_{_{0}}}}}}(o_{_{\nu_{_{\tau-1-\mu_{_{0}}}}}},i,\tau-1-\mu_{_{0}},\nu_{_{\tau-1-\mu_{_{0}}}}),\ldots,q_{_{\nu_{_{\tau-1}}}}(o_{_{\nu_{_{\tau-1}}}},i,\tau-1,\nu_{_{\tau-1}}),$$
		 be $\mu_{_{0}}=t_{_{s}}+m$ of its consecutive predecessors (in the $i$th session)  with the following vector of answers,
		 $$\varrho^b=(\ciphervec^{b}(\nu_{_{\tau-1-\mu_{_{0}}}}),\ldots,\ciphervec^{b}(\nu_{_{\tau-1}})).$$ 
		Let $\xi' \isdef (\ciphervec_{_{\nu_{_{\tau-1-\mu_{_{0}}}}}}, \cdots, \ciphervec_{_{\nu_{_{\tau-1}}}}) $. We use well-founded induction to prove that for any $\nu$ we have 
		$$  Pr_{0}(\varrho^0=\xi'|\overline{Col_{_{\nu}}}) = Pr_{1}(\varrho^1=\xi'|\overline{Col_{_{\nu}}}).$$
		Since the distribution of the initial states, and consequently, the answers to the first queries are uniformly distributed the base of the well-founded induction holds.
		
		For induction step, note that,
		$$ Pr_{b,\overline{Col_{_{\nu}}}}(\varrho^b=\xi')=Pr_{b,\overline{Col_{_{\nu}}}}(\varrho^b=\xi'|\xi' \in H_{_{\nu}}(\xi))
		Pr_{b,\overline{Col_{_{\nu}}}}(\xi' \in H_{_{\nu}}(\xi))+$$
		$$ Pr_{b,\overline{Col_{_{\nu}}}}(\varrho^b=\xi'|\xi' \not \in H_{_{\nu}}(\xi))
		Pr_{b,\overline{Col_{_{\nu}}}}(\xi' \not \in H_{_{\nu}}(\xi)) $$
		First, note that the event $\xi' \in H_{_{\nu}}(\xi)$ does not depend on $b$, since the size of the set $H_{_{\nu}}(\xi)$ does not depend on $b$.		
		On the other hand, $\xi' \in H_{_{\nu}}(\xi)$ implies the event $Col_{_{\nu}}$, indicating that the first term in the sum is equal to zero. Also, 
		the union of $\overline{Col_{_{\nu-1}}}$ and $\xi' \not \in H_{_{\nu}}(\xi)$ is equal to $\overline{Col_{_{\nu}}}$, hence
		$$ Pr_{b,\overline{Col_{_{\nu}}}}(\varrho^b=\xi'|\xi' \not \in H_{_{\nu}}(\xi))= Pr_{b}(\varrho^b=\xi'|\overline{Col_{_{\nu-1}}}),$$
		 that does not depend on $b$ by induction hypothesis.
		
		{\bf (b) $\Rightarrow$} For this case, we use well-founded induction to prove that for any $\nu$ we have $Pr_{0}(Col_{_{\nu}})=Pr_{1}(Col_{_{\nu}})$. Note that the equality 
		is trivially true for the minimal elements (corresponding to the initial states). 
		
		Within the same setting as in Part$(a)$, for any $b \in \{0,1\}$ we have,
		$$ Pr_{b}(Col_{_{\nu}})=Pr_{b}(Col_{_{\nu}}|Col_{_{\nu-1}}) Pr_{b}(Col_{_{\nu-1}})+Pr_{b}(Col_{_{\nu}}|\overline{Col_{_{\nu-1}}}) Pr_{b}(\overline{Col_{_{\nu-1}}}).$$
		Since $Pr_{b}(Col_{_{\nu}}|Col_{_{\nu-1}})=1$, by induction the first term of the sum is independent of $b$. Hence, by induction, it suffices to prove that 	$Pr_{b}(Col_{_{\nu}}|\overline{Col_{_{\nu-1}}})$ is independent of $b$.	But,
		$$Pr_{b}(Col_{_{\nu}}|\overline{Col_{_{\nu-1}}})=\displaystyle{\sum_{\xi'}} Pr_{b}(Col_{_{\nu}}|\overline{Col_{_{\nu-1}}} \ \& \ \varrho^b=\xi') Pr_{b}(\varrho^b=\xi'|\overline{Col_{_{\nu-1}}}).$$
		In each term, The second component is independent of the bit $b$ by the proof of part $(a)$. For the first term, 
		note that the probability only depends on the vectors 
		$\ciphervec^{b}(\nu_{_{\tau+1}})$
		satisfying Equation~\ref{eq:psi3}. But, since the initial states are uniformly chosen random vectors, the matrix $\Wmat$ and for any $j \in [ \ell ]$, the matrices  $\Lmat_{j}, \Fmat_{j}$ are uniformly chosen random invertible matrices,
		and $\BigPi_{_{\key}}$ is a uniformly chosen random permutation, the distributions of vectors
		$\Fmat_{\sigma(h)}^{-1} (\ciphervec(h))$,
		$\Fmat_{\sigma(h)}^{-1}\ciphervec(h)$, 
		$\Wmat\memvec(h)$,
		$\Wmat\widehat{\memvec}(h+1)$,
		$\Lmat\state_{i}(\tau)$,
		$\Lmat\hstate_{i}(\tau+1)$,
		and 
		$\Fmat\BigPi_{_{\key}}(\plainvec^b_{_{\nu_{_{\tau}}}}(\tau))$
		are the same, and consequently, by Equations~\ref{eq:psi1}, \ref{eq:psi2} and \ref{eq:psi3}
		the probability $Pr_{b}(Col_{_{\nu}}|\overline{Col_{_{\nu-1}}} \ \& \ \varrho^b=\xi')$ does not depend on $b$ (see Section~\ref{sec:conclusion} for a discussion on this part).
		
		{\bf (c) $\Rightarrow$} For this part note that,
			$$Pr_{b}(Col) \leq \displaystyle{\sum_{\nu=2}^{r}}  Pr_{b}(Col_{_{\nu}}|\overline{Col_{_{\nu-1}}}).$$
			But considering Part~$(b)$ we know that $\ciphervec_{_{\nu}}$ is uniformly distributed,
			$$Pr_{b}(Col) \leq \displaystyle{\sum_{\nu=2}^{r}}  Pr_{b}(Col_{_{\nu}}|\overline{Col_{_{\nu-1}}}) \leq 
			\displaystyle{\sum_{\nu=2}^{r}} \frac{2(\nu -1)}{q^n} \leq 2r^2 2^{-k},$$
			which is a negligible function of $k$ since $r$ is bounded by a polynomial function of $k$.
		
}\end{proof}	

Clearly, using Proposition~\ref{thm:main} one may prove the main security result as follows.

\begin{thm} \label{thm:thm3}
The self-synchronized stream cipher $\matr{S}_{\sigma}^4(\mathfrak{P})$ is LORBACPA$^+$ secure.
\end{thm}
\begin{proof}{By Proposition~\ref{pro:randomcase} it suffices to prove the claim for $\matr{S}_{\sigma}^4(\mathfrak{U})$. For this we have
		
		{\scriptsize	\[\begin{array}{ll}
			Adv^{LORBACPA^+}_{ \mathcal{A},{\matr{S}_{\sigma}^4(\mathfrak{U})}}(k)&
			=Pr_{1}(output({\mathcal A})=1)-Pr_{0}(output({\mathcal A})=1)\\&=Pr_{1}(output({\mathcal A})=1\mid Col^1)Pr_{1}(Col^1)+Pr_{1}(output({\mathcal A})=1\mid \overline{Col^1})Pr_{1}(\overline{Col^1})\\&-Pr_{0}(output({\mathcal A})=1\mid Col^0)Pr_{0}(Col^0)-Pr_{0}(output({\mathcal A})=1\mid \overline{Col^0})Pr_{0}(\overline{Col^0})
			\\&=\left(Pr_{1}(output({\mathcal A})=1\mid Col^1)-Pr_{0}(output({\mathcal A})=1\mid Col^0)\right) \negl_{_{0}}(k)\\&\leq \negl'(\key).
			\end{array}
			\]}
		
		Consequently, we obtain
		\[\begin{array}{ll}
		\Adv^{LORBACPA^+}_{ \mathcal{A},{\matr{S}_{\sigma}^4(\mathfrak{P})}}(k) &\leq\Adv^{LORBACPA^+}_{ \mathcal{A},{\matr{S}_{\sigma}^4(\mathfrak{U})}}(k) + \negl'(k)\\&\leq  \negl(k).
		
		\end{array}\]
}\end{proof}

\section{Concluding remarks}\label{sec:conclusion}

In this article we introduced the security model LORBACPA$^+$ for self-synchronized stream ciphers which is stronger than the traditional blockwise
LOR-IND-CPA, and based on contributions of G.~Mill\'{e}rioux et.al., we introduced a new self-synchronized stream cipher $\matr{S}_{\sigma}^4(\mathfrak{P})$
which is secure in this stronger model. It is instructive to note that the main idea giving rise to this stronger security property is the fact that in the new setup
which is based on control theoretic unknown input observer design for the receiver, one is able to use totally random initial state vectors for encryption.

It is also interesting to have a control-theoretic view to our security proof as an {\it uncontrolability} result. From this point of view, an adversary is a stochastic discrete 
dynamical system gaining information from the answers it receives to its queries as control-inputs. Hence, the aim of the adversary is to make collisions for the answers since 
having no collision gives rise to a {\it no-information state} because of the uniform distribution of the answers (as a consequence of the uniform distribution of the initial states of the oracles and system secret parameters). This can be thought of as a game in which the adversary tries to make collisions while in the state of a self-avoiding walk (i.e. no collision) of this dynamical system, the adversary gains no information about $b$. Therefore, the whole scenario is to control the system for the objective of maximizing the probability of a collision (hopefully to become noticeable), where from this point of view, {\it security} can be interpreted as an {\it uncontrolability} property. In this setting, a simple intuition supporting our security proof is based on the facts that  being able to choose i.i.d random vectors for the initial states of the oracles guaranties that different oracles have independent trajectories, while the fact that the length of runs (i.e. walks) of each oracle is bounded by a polynomial function of the security parameter (since the adversary is a polynomial-time algorithm) and the fact that each step of the run as an $n$-dimensional vector has an exponential number of possibilities, makes sure that the probability of a collision is bounded by a negligible function of the security parameter.

Let us also add a couple of comments on practical issues. First, note that Part~$(b)$ of Proposition~\ref{thm:main} is still valid even if only the matrix $\Lmat$ and the permutation $\BigPi_{_{\key}}$ are secretly chosen uniformly at random,  however, since in practice, and in particular for small parameters used in a lightweight setting, 
real simulations fail to completely satisfy theoretical assumptions we have also added the matrices $\Fmat$ and $\Wmat$ to the secret parameters of the ciphersystem to ensure uniform randomization mixing (see Equations~\ref{eq:psi1} and \ref{eq:psi2}). On the other hand, we would like to mention that from a practical point of view, our ciphersystem is as close as to a CCA secure streamcipher while it still is error-resistant and self-synchronized, which in our opinion, makes it more applicable in comparison to a CCA secure stream cipher along with a synchronizer, at least in  noisy environments.

Naturally, more practical issues and applications concerning the implementation of this new streamcipher should be the subject of further investigations.

 %----------------------------------------------------section{Bibliography}-------------

\end{document}